\documentclass[format=acmsmall, review=false, screen=true]{acmart}

\usepackage{booktabs} 

\usepackage[ruled]{algorithm2e} 

\SetAlFnt{\small}
\SetAlCapFnt{\small}
\SetAlCapNameFnt{\small}
\SetAlCapHSkip{0pt}
\IncMargin{-\parindent}

\usepackage{amsmath}
\usepackage{amssymb}
\usepackage{amsfonts}
\usepackage{amsthm}
\usepackage{enumerate}
\usepackage{esint}
\usepackage{etoolbox}
\usepackage{graphicx}
\usepackage{url}
\usepackage{verbatim}

\newcommand{\repthanks}[1]{\textsuperscript{\ref{#1}}}
\makeatletter
\patchcmd{\maketitle}{\def\thanks}{\let\repthanks\repthanksunskip\def\thanks}{}{}
\patchcmd{\@maketitle}{\def\thanks}{\let\repthanks\@gobble\def\thanks}{}{}
\newcommand\repthanksunskip[1]{\unskip{}}
\makeatother

\newcounter{theoremXXX}
\newtheorem{property_XXX}[theoremXXX]{Property}{\bfseries}{\itshape}
\newtheorem{remark}{Remark}{\itshape}{\rmfamily}
{\itshape}{\rmfamily}
\newtheorem*{claim*}{Claim}{\itshape}{\rmfamily}
\newtheorem*{remark*}{Remark}{\itshape}{\rmfamily}
\newtheorem*{defn*}{Definition}{\bfseries}{\rmfamily}
\newtheorem*{example*}{Example}{\bfseries}{\rmfamily}

\begin{document}

\title{Tight Bounds on Online Checkpointing Algorithms}

\author{Achiya Bar-On}
\affiliation{%
\institution{Department of Mathematics, Bar-Ilan University}
\city{Ramat-Gan}
\country{Israel}
}
\email{abo1000@gmail.com}
\author{Itai Dinur}
\affiliation{%
\institution{Computer Science Department, Ben-Gurion University}
\city{Beer-Sheba}
\country{Israel}
}
\email{}
\author{Orr Dunkelman}
\affiliation{%
\institution{Computer Science Department, University of Haifa}
\city{Haifa}
\country{Israel}
}
\email{orrd@cs.haifa.ac.il}
\author{Rani Hod}
\affiliation{%
\institution{Department of Mathematics, Bar-Ilan University}
\city{Ramat-Gan}
\country{Israel}
}
\email{rani.hod@math.biu.ac.il}
\author{Nathan Keller}
\affiliation{%
\institution{Department of Mathematics, Bar-Ilan University}
\city{Ramat-Gan}
\country{Israel}
}
\email{nkeller@math.biu.ac.il}
\author{Eyal Ronen}
\affiliation{%
\institution{Computer Science Department, The Weizmann Institute}
\city{Rehovot}
\country{Israel}
}
\email{eyal.ronen@weizmann.ac.il}
\author{Adi Shamir}
\affiliation{%
\institution{Computer Science Department, The Weizmann Institute}
\city{Rehovot}
\country{Israel}
}
\email{adi.shamir@weizmann.ac.il}

%

\begin{abstract}
The problem of online checkpointing is a classical problem with numerous
applications which had been studied in various forms for almost $50$ years. In the simplest version of this problem, a user has to maintain $k$ memorized checkpoints during a long computation, 
where the only allowed operation is to move one of the checkpoints
from its old time to the current time, and his goal is to keep the
checkpoints as evenly spread out as possible at all times.

Bringmann et al.~studied this problem as a special case of an online/offline
optimization problem in which the deviation from uniformity is measured by the natural discrepancy metric of the
worst case ratio between real and ideal segment lengths. They showed this discrepancy is smaller than $1.59-o(1)$ for all~$k$, and smaller than $\ln4-o(1)\approx1.39$ for the
sparse subset of $k$'s which are powers of~2. In addition, they obtained upper bounds on the achievable discrepancy for some small values of~$k$.

In this paper we solve the main problems  left open in the 
above-mentioned
paper by proving that $\ln4$ is a tight
upper and lower bound on the asymptotic discrepancy for all large $k$, and by providing tight
upper and lower bounds (in the form of provably optimal checkpointing algorithms, some of which are 
in fact better than those of Bringmann et al.) for all the small values of $k \leq 10$.

In the last part of the paper we describe some new applications of this online checkpointing problem.
\end{abstract}


\maketitle
\renewcommand{\shortauthors}{A.~Bar-On et al.}


\section{\label{sec:Intro}Introduction and Notation}

Most programs perform some irreversible operations, and thus they
can only be run in a forward direction. However, in many cases we
would like to roll back a computation to an earlier point in time.
When the computation is short, we can just rerun the computation from
the beginning, but when the computation requires many days, a better
strategy is to memorize several copies of the full state of the computation
at various times. These memorized states (called \textit{checkpoints})
make it possible to roll the computation back from time $T$ to any
earlier time $T'<T$ by restarting the computation from the last available checkpoint
which was memorized before $T'$. This checkpointing technique is
extremely useful in many real life applications: For example, when
we want to interactively debug a new program we may want to randomly
access earlier points in the execution in order to find the source
of a problem; in fault tolerant computer systems we may want to undo
the effects of a faulty hardware; and during lengthy simulations of physical
systems we may want to explore the effect of changing some parameter,
such as the temperature, at some earlier point in time without rerunning
the simulation from the beginning.

In principle, we can try to memorize the full state of the computation
after each step, but for long computations this requires an unrealistic
amount of memory. Instead, we assume that we have some bounded amount
of memory which suffices to keep $k$ checkpoints. At time $T$, these
checkpoints are spread within the time interval $[0,T]$, dividing
it into $k+1$ subintervals between consecutive checkpoints (where
the endpoints $0$ and $T$ can be viewed as virtual checkpoints which
require no additional memory). As $T$ increases, the last subinterval
gets longer, and at some point we may want to relocate one of the
old checkpoints by reusing its memory to store the current state of
the computation. A checkpointing algorithm can thus be viewed as an
infinite pebbling game in which we place $k$ pebbles on the positive
side of the time axis, and then repeatedly perform update operations
which move one of the pebbles to the right of all the other pebbles.

The first paper dealing with this problem seems to be ``Rollback and Recovery Strategies for Computer Programs''~\cite{CR72}, published in 1972, while the first paper which tried to solve it optimally was ``On the Optimum Checkpoint Interval''~\cite{Gelenbe79}, published in 1979. Over the years, dozens of academic research papers were published in this area, most notably~\cite{TB84} in 1984,~\cite{BGRS94} in 1994, and~\cite{APS13,BDNS13} in 2013. However, many of these papers either dealt with concrete
applications of the problem in other areas (especially in distributed computing where the notion of a timeline is different), or used other optimization criteria (which make their optimal solutions incomparable with ours). The mathematical problem we are dealing with in this paper was mentioned in~\cite{APS13} and studied in~\cite{BDNS13}, and we closely follow their model and notation.

At any time $T$, we define a \textit{snapshot} as the ordered sequence 
of current checkpoint locations $S=(T_{1},\ldots,T_{k})$. Within each snapshot,
we refer to the checkpoints by their \textit{freshness index}~$p$,
where checkpoint~$1$  stores the oldest state and checkpoint~$k$ stores the newest state. Starting
from an initial snapshot $S_{k}=\left(t_{1},t_{2},\ldots,t_{k}\right)$,
we define for every $i\ge k+1$ the $i$-th \textit{update action} as a pair $(t_{i},p_{i})$ in which $p_{i}$
is the freshness index of the checkpoint whose memory we want to reuse by moving it to time $t_{i}$.
A typical example of how one snapshot is transformed into another
snapshot by an update operation is described in Fig.~\ref{fig:transition}.
The effect of the $i$-th update action is to unify the two consecutive
subintervals which were separated by the $p_{i}$-th oldest active
checkpoint at time $T=t_{i}$, and to create a new subinterval which
ends at $t_{i}$. Note that with this notation, each update action
affects multiple freshness indices within the snapshot; in particular,
the freshness index of active checkpoint $T_{j}$ for $1\le j<p_{i}$
is left unchanged, and is decreased by one for $p_{i}<j\le k$. To
demonstrate this point, consider a sequence of updates in which $p_{i}=1$
for all $i\ge k+1$: it updates the $k$ memory locations in a round
robin way since it always updates the oldest active checkpoint by
overwriting it with the newest checkpoint, shifting all freshness
indices by one. On the other hand, a sequence of updates in which
$p_{i}=k$ for all $i\ge k+1$ keeps updating the same memory location,
pushing its associated checkpoint further and further to a later time, with no change to the other checkpoints.

\begin{figure}[h]
\centering{}\includegraphics[clip]{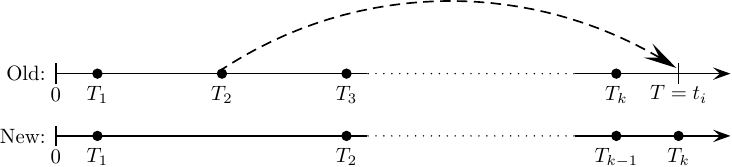}
\caption{\label{fig:transition}Transition from old to new snapshot for the
update action $\left(t_i,2\right)$.}

\end{figure}


In this model, the time complexity of rolling back a computation from
time $T$ to time $T'$ is assumed to be proportional to the distance
between $T'$ and the last checkpoint that precedes $T'$ in the snapshot
at time $T$, and thus its worst case happens when we decide to roll
back to just before the end of the longest subinterval. A \textit{checkpointing
algorithm} $(t,p)$ consists of 
a monotonically increasing and unbounded sequence of update times $t=\{t_{i}\}_{i=1}^\infty$ and a pattern sequence $p=\{p_{i}\}_{i=k+1}^\infty$, forming an initial snapshot and an infinite sequence of update actions; its goal is to make the length of this longest subinterval as short as possible. Clearly, no checkpointing algorithm
can make this length shorter than the subinterval length in a perfectly
uniform partition of $[0,T]$, which is $T/(k+1)$.
We say that a snapshot $S=\left(T_{1},\ldots,T_{k}\right)$
of a $k$-checkpoint algorithm $\textsc{Alg}=\left(t,p\right)$ is
$q$-\emph{compliant} at time $T$ if the
$k+1$ subintervals defined by $S$ satisfy $T_{j}-T_{j-1}\le qT/\left(k+1\right)$
for $j=1,\ldots,k+1$,\footnote{We write $T_{0}=0$ and $T_{k+1}=T$ for convenience.} and that 
\textsc{Alg} is \textit{$q$-efficient} if its snapshots are $q$-compliant at all times $T\ge t_k$.
Finally, the \emph{efficiency} of a checkpointing algorithm is defined as the smallest $q$ for
which it is $q$-efficient. 

Notice that the problem of efficient checkpointing
can be viewed as a special case of an online/offline optimization
problem: If we knew in advance the time $T$ at which we would like
to roll back the computation, we could make each subinterval as small
as $T/(k+1)$. However, in the online version of the problem, we do
not know $T$ in advance, and thus we have to position the checkpoints
so that they will be roughly equally spaced at all times. The efficiency
of the solution is the ratio between what we can achieve in the online
and offline cases, respectively, and the goal of the \textit{online
checkpointing problem} is to find the smallest possible
efficiency $q_{k}$ achievable by the best $k$-checkpoint algorithm for any
given $k$.

Clearly, $q_k\ge1$ for all $k\ge2$, and cannot be too close to $1$ since any snapshot in which all subintervals have 
roughly the same length will be transformed by the next update operation to a snapshot in which one of the subintervals will be the union of two previous
subintervals, and thus will be about twice as long as the other subintervals.\footnote{%
Actually $q_k \ge (k+1)/k$ since subinterval $k+1$ has zero length upon updating, as noted in~\cite[Theorem~3]{APS13}.}
On the other hand, there is a very simple subinterval doubling algorithm from~\cite[Section~3.1]{APS13} which is 2-efficient:
Assuming WLOG that $k$ is even, the algorithm starts with the snapshot $(1,2,3, \ldots , k)$, and performs 
the sequence of update actions $(k+2, 1), (k+4, 2), \ldots , (2k, k/2)$, yielding the snapshot 
$(2,4,6, \ldots, 2k)$. Since this snapshot is the same as the original snapshot up to a scaling factor of~2, we can continue with update actions $(2k+4, 1), (2k+8, 2), \ldots , (4k, k/2)$ and so on. This is a \textit{cyclic} algorithm, repeating the same sequence of freshness indices again and again but with times which form a geometric progression. As in each snapshot there are only
two possible lengths for the subintervals of the form $x$ and $2x$, all the snapshots in this algorithm are 
2-compliant, and thus the algorithm is 2-efficient.

The best strategy for keeping the checkpoints as uniform as possible at all times is thus to keep in 
each snapshot a variety of subinterval lengths, so that the algorithm will always be able to join two 
relatively short adjacent subintervals into a single subinterval which is not too long. This can be 
viewed as a generalization of the algorithm that creates Fibonacci numbers: Whereas the standard algorithm 
is always adding the last two numbers and placing their sum on the right, in our case we can add any two
consecutive numbers in the sequence, replacing them by their sum and adding any number we want on the right.
Analyzing this problem is surprisingly difficult, and so far there had been no tight bounds on the best possible
efficiencies $q_k$ of online checkpointing algorithms in this model. 


The main results in~\cite{BDNS13} are two online checkpointing algorithms whose
asymptotic efficiencies are $\ln4+o(1) \approx 1.39$ for the sparse subset of $k$'s which are powers of 2,
and $1.59$ for general~$k$. In addition, they proved in their model the first nontrivial asymptotic lower bound 
of~$2-\ln2-o(1)\approx1.30$. However, since the upper and lower bounds did not match, 
it was not clear whether the checkpointing algorithms they proposed were asymptotically optimal. 
For small values of $k<60$ they presented concrete checkpointing algorithms 
whose efficiencies were all below $1.55$, but again it was not clear whether they were optimal.

In this paper we solve the main open problems related to the mathematical formulation of the
problem which was defined and studied in~\cite{BDNS13}. In particular, we
develop a new checkpointing algorithm with an asymptotic
efficiency of $\ln4$ for all values of $k$, and prove its optimality by providing a matching asymptotic lower bound. 
For all the small values of $k<10$ we develop optimal checkpointing algorithms by proving tight upper and lower bounds 
on the achievable efficiency for these $k$'s. This analysis enables us to show that for some values of $k$ (such as $k=8$), 
the algorithms presented in~\cite{BDNS13} are in fact suboptimal.  

\medskip{}
The rest of this paper is organized as follows.
In Section~\ref{sec:basic} we go over basic observations about checkpointing algorithms (some from~\cite{BDNS13}, some new).
In Section~\ref{sec:small-k} we focus on moderately small values of $k$ and provide optimal algorithms for $k\le 10$.
In Section~\ref{sec:recursive-construction} we construct a recursive algorithm of asymptotically optimal efficiency $\ln4+o(1)$.
In Section~\ref{sec:lower-bound} we prove a matching asymptotic lower bound of $\ln4-o(1)$.
In Section~\ref{sec:applications} we present two new applications of online checkpointing.
In Section~\ref{sec:conclusion} we provide concluding remarks.


\section{\label{sec:basic}Basic Observations}

By definition, a $k$-checkpoint $\textsc{Alg}=(t,p)$ is $q$-efficient if and only if its snapshots
at all times $T\ge t_{k}$ are $q$-compliant. However, as noted in~\cite[Lemma~2]{APS13} (and also~\cite[Lemma~1]{BDNS13}),
it suffices to verify compliance only at the discrete times $T\in\left\{ t_{i}\right\} _{i=k}^{\infty}$.
It makes sense thus to only consider ``standard'' snapshots $S_{i}$
taken at time $t_{i}$ for $i\ge k$. Moreover, as shown in~\cite[Lemma~2]{BDNS13},
besides compliance of the initial snapshot $S_{k}$, it suffices to
verify compliance of just two subintervals of $S_{i}$ for every $i>k$
\textemdash{} subinterval $k$, that ends in the new checkpoint~$t_{i}$,
and subinterval $p_{i}$, created by merging two consecutive subintervals.

The following two observations about the sequence $p=\left(p_{i}\right)_{i=k+1}^{\infty}$
were mentioned in~\cite[Section~6]{BDNS13} without proof.

\begin{property_XXX}\label{property:update-1-infinitely-often}
Without loss of generality
we can assume a $k$-checkpoint algorithm updates the least recent
checkpoint infinitely often (i.e., ${\displaystyle \liminf_{i\to\infty}p_{i}=1}$).
\end{property_XXX}
\begin{proof}
Fix a $q$-efficient algorithm $\textsc{Alg}=\left(t,p\right)$ and
consider its standard snapshots $S_{i}=\left(T_{1}^{i},\ldots,T_{k}^{i}\right)$
for $i\ge k$. For notational convenience let $\Delta_{j}^{i}=T_{j}^{i}-T_{j-1}^{i}$
for $i\ge k$ and $j=1,\ldots,k$. The update times sequence $t$
is unbounded and $t_{i}=T_{k}^{i}=\sum_{j=1}^{k}\Delta_{j}^{i}$,
so there must exist some minimal $J<k$ for which the sequence $\left(\Delta_{J}^{i}\right)_{i=k}^{\infty}$
is unbounded. If $J=1$ we are done; otherwise we show how to modify
the algorithm, while maintaining $q$-efficiency, such that $\left(\Delta_{J-1}^{i}\right)_{i=k}^{\infty}$
would be unbounded as well.

Let $M=\sup_{i}\Delta_{J-1}^{i}$ and pick $I$ such that $\Delta_{J}^{I}>M$.
No further update action $\left(p_{i},t_{i}\right)$ can update a
checkpoint of freshness index $p_{i}<J$ since that would result in
$\Delta_{J-1}^{i}>M$. In particular, we have $\Delta_{J}^{i}>M$
for all $i\ge M$. Pick $i>I$ such $p_{i}=J$ and $\Delta_{J+1}^{i-1}>M$.
It is possible since $\left(\Delta_{J+1}^{i}\right)_{i=M}^{\infty}$
is unbounded and cannot decrease unless checkpoint $J$ is updated.
We modify the algorithm to update checkpoint $J-1$ instead of $J$
at time $t_{i}$, that is, set $p_{i}=J-1$. Instead of creating a
$q$-compliant subinterval of length
\[
\Delta_{J}^{i}=\Delta_{J}^{i-1}+\Delta_{J+1}^{i-1}>\Delta_{J}^{i-1}+M,
\]
this update action now creates a subinterval of length
\[
\Delta_{J-1}^{i}=\Delta_{J-1}^{i-1}+\Delta_{J}^{i-1}\le\Delta_{J}^{i-1}+M,
\]
which is still $q$-compliant. The algorithm remains $q$-efficient
since future updates only touch subintervals $j\ge J$ and $\Delta_{J}^{i}$
did not grow. This process can be repeated as long as $\sup_{i}\Delta_{J-1}^{i}$
remains finite. Note that in the limiting algorithm, checkpoint $J-1$
is updated infinitely often so $\sup_{i}\Delta_{J-1}^{i}$ cannot
be finite; thus $\left(\Delta_{J-1}^{i}\right)_{i=k}^{\infty}$ must
be unbounded.
\end{proof}

\begin{remark}\label{rem:rebase}
An important consequence of Property~\ref{property:update-1-infinitely-often}
is that we can essentially ignore the compliance of the initial snapshot
$S_{k}$ by \emph{rebasing}, i.e., running the algorithm until all
checkpoints present in $S_{k}$ are overwritten and treating the then-current snapshot as the new initial $\left(t_{1},\ldots,t_{k}\right)$.
\end{remark}

\begin{property_XXX}\label{property:never-update-k}
Without loss of generality we can assume a $k$-checkpoint algorithm never updates the most recent checkpoint (i.e., $p_{i}<k$ for all $i\ge1$).
\end{property_XXX}
\begin{proof}
Fix a $q$-efficient algorithm $\textsc{Alg}=\left(t,p\right)$.
If $p_{i+1}=k$ then the $(i+1)$-th update transforms the snapshot
$S_{i}=\left(T_{1},\ldots,T_{k-1},t_{i}\right)$ to
$S_{i+1}=\left(T_{1},\ldots,T_{k-1},t_{i+1}\right)$. 
We have $t_{i+1}\le T_{k-1}G$ since subinterval~$k$ in $S_{i+1}$ is $q$-compliant, 
so for any $t_{i}\le T\le t_{i+1}$ we have $T-T_{k-1}\le qT/\left(k+1\right)$.
Hence we can simply skip the $i$-th update action $\left(p_{i},t_{i}\right)$
altogether and the algorithm remains $q$-efficient.
\end{proof}
\begin{remark*}
Property~\ref{property:never-update-k} means that
the two last checkpoints in snapshot $S_{i}$ are $t_{i-1}$ and $t_{i}$,
and thus subinterval $k$ is $q$-compliant if and only if $t_{i-1}-t_{i}\le qt_{i}/\left(k+1\right)$,
that is, $t_{i}\le G\cdot t_{i-1}$, where $G=G\left(q\right):=\left(k+1\right)/\left(k+1-q\right)$.
We refer to this condition by saying that the update times sequence
$t=\left(t_{i}\right)_{i=1}^{\infty}$ should be $G$-\emph{subgeometric}.
\end{remark*}

Next we introduce the notion of \emph{cyclic} algorithms. Upper bounds
on $q_{k}$ presented in this paper, as well as in~\cite{APS13,BDNS13},
are all achieved by cyclic algorithms. Given a positive integer $n$
and a real number $\gamma>1$, a $k$-checkpoint algorithm $\textsc{Alg}=\left(t,p\right)$
is $\left(n,\gamma\right)$-\emph{cyclic} if $t_{n+i}=\gamma\cdot t_{i}$
for all $i\ge1$ and $p_{i}=p_{n+i}$ for all $i\ge k+1$. It has
been observed in~\cite[Lemma 5]{BDNS13} that any $q$-efficient
$\left(n,\gamma\right)$-cyclic algorithm must satisfy $\gamma\le G^{n}$
(to see this, apply subgeometry~$n$ times). An $\left(n,\gamma\right)$-cyclic
algorithm is called $\left(n,G\right)$\emph{-geometric} when $\gamma=G^{n}$
(and thus $t_{i+1}=G\cdot t_{i}$ for $i\ge k$).

\medskip{}

We finish this section with two observations about the exponential
growth of update times in efficient algorithms, relevant for upper
and lower bounds on $q_{k}$.

The first one is an improvement of~\cite[Lemma 8]{BDNS13}:
\begin{property_XXX}\label{property:supergeometric}
Any $q$-efficient $k$-checkpoint algorithm $\textsc{Alg}=\left(t,p\right)$ satisfies, without loss of generality, $t_{i+2}>t_{i}\cdot G$ for all $i\ge k$.
\end{property_XXX}
\begin{proof}
Starting with a $q$-efficient algorithm, we consider checkpoint updates
in their natural order and modify the algorithm, while maintaining efficiency,
such that the property holds. At each step, the only modifications
we make are to skip or delay an update, which ensures that $\left(t_{i}\right)_{i=1}^{\infty}$
is still unbounded. The basic simple idea in all modifications is
that if an update that merges two subintervals is possible at some time
$T$, i.e., the newly created subinterval is $q$-compliant, then an
update that merges the same subintervals is also possible at any time
$T'>T$.

If the property does not hold, consider the smallest $i$ for which
$t_{i+2}\le t_{i}\cdot G$. We show we can either skip one of the
updates at times $t_{i+1}$, $t_{i+2}$ or delay the update at time
$t_{i+2}$ until time $t_{i+1}\cdot G>t_{i}\cdot G$. For $j=1,2$,
denote by $x_{j}$ the next time at which the checkpoint at $t_{i+j}$
is updated, and by $\ell_{j}$ the label of the actual checkpoint updated.\footnote{In contrast to a temporary index of the checkpoint in the checkpoint sequence at some snapshot $S$, a label of a checkpoint is fixed.} Note that $\ell_{1}\neq\ell_{2}$ by Property~\ref{property:never-update-k},
and that the order of $x_{1}$ and $x_{2}$ is undetermined.

If $x_{1}<x_{2}$ (the checkpoint $\ell_{2}$ at $t_{i+2}$ is in place
by the time the checkpoint $\ell_{1}$ at $t_{i+1}$ is updated), then
eliminate the checkpoint update at $t_{i+1}$ (keeping the sequence $G$-subgeometric
as $t_{i+2}\le t_{i}\cdot G$) and switch roles between the two checkpoints
throughout the rest of the algorithm. Namely, update $\ell_{2}$ from
its previous update (before being updated at $t_{i+2}$) directly
at time $x_{1}$,\footnote{If checkpoint $\ell_{2}$ was firstly used at $t_{i+2}$, then simply
use $\ell_{2}$ first at $x_{1}$ in the modified algorithm.} and update $\ell_{1}$ from its previous update directly at time
$t_{i+2}$ (and then again at $x_{2}$ and so forth).

Otherwise, $x_{2}<x_{1}$; consider the time interval $(t_{i+1},t_{i+1}\cdot G]$
and denote by $\ell_{3}$ the checkpoint that is removed from it at the
latest time before $x_{1}$ (note that there is at least one such
checkpoint, namely $\ell_{2}$, hence it may be that $\ell_{3}=\ell_{2}$).
Delay the update of $\ell_{3}$ to time $t_{i+1}\cdot G$ and skip
all other updates in $(t_{i+1},t_{i+1}\cdot G]$ that are removed
from this time interval before $x_{1}$. In other words, for each checkpoint
(except $\ell_{3}$) updated in the time interval $(t_{i+1},t_{i+1}\cdot G]$
and updated again in the time interval $(t_{i+1}\cdot G,x_{1})$, update
it directly from its previous update (before $t_{i+1}$) to its next
one (after $t_{i+1}\cdot G$).
\end{proof}

An immediate corollary of Property~\ref{property:supergeometric} is that
$t_{i+j}>t_{i}\cdot G^{\left\lfloor j/2\right\rfloor }$ for $j\ge0$;
in particular, $t_{i+j}>t_{i}\cdot G^{2}$ for $j\ge4$. Our next
observation says when we can get $t_{i+3}>t_{i}\cdot G^{2}$.
\begin{property_XXX}\label{property:supergeometric2}
Let $S=\left(T_{1},\ldots,T_{k}\right)$ be a snapshot of some $q$-efficient $k$-checkpoint algorithm $\textsc{Alg}=\left(t,p\right)$
such that $T_{j}=t_{i}$ and $T_{j+1}=t_{i+3}$ for some $j=1,\ldots,k-1$.
Thus, without loss of generality, $t_{i+3}>t_{i}\cdot G^2$.
\end{property_XXX}
\begin{proof}
Starting with a $q$-efficient algorithm that satisfies Property~\ref{property:supergeometric},
we modify it such that Property~\ref{property:supergeometric} holds
as well. As in the previous proof, for $j=1,2$, denote by $\ell_{j}$
the label of the checkpoint at $t_{i+j}$ and denote by $x_{j}$ the time
by which it is removed from the time interval $\left(t_{i},t_{i+3}\right)$.

We show that if $t_{i+3}\leq t_{i}\cdot G^{2}$ then we can delay
the update at time $t_{i+1}$ until $t_{i}\cdot G$ and eliminate
the update at time $t_{i+2}$. Note that $t_{i+1}\leq t_{i}\cdot G<t_{i+2}$
by Property~\ref{property:supergeometric}, hence after this change Property~\ref{property:supergeometric} continues to hold.

If $x_{1}<x_{2}$ we switch the roles of the two checkpoints by updating
$\ell_{2}$ directly at $x_{1}$ and $\ell_{1}$ at $t_{i}\cdot G$
and again at $x_{2}$; otherwise, $x_{2}<x_{1}$ and we simply delay
the update of $\ell_{1}$ to time $t_{i}\cdot G$ and skip the update
at $t_{i+2}$ by updating $\ell_{2}$ directly at $x_{2}$ from its
previous update.
\end{proof}


\section{\label{sec:small-k}Optimal Algorithms for Small Values of $k$}

\subsection{Round-robin and $k\le5$}

We now analyze the efficiency of the \textsc{Round-Robin} algorithm,
which is geometric and always updates the oldest checkpoint (i.e.,
$p_{i}=1$ for all $i\ge k+1$).\footnote{The case $k=3$ of \textsc{Round-Robin} was considered in~\cite[Theorem~1]{BDNS13}
under the name \textsc{Simple}.} Besides serving as a first example, \textsc{Round-Robin} is optimal
for $k\le3$ and will make an appearance within the asymptotically
optimal algorithm \textsc{Recursive} of Section~\ref{sec:recursive-construction}.
\begin{proposition}\label{prop:RR}
The efficiency of $k$-checkpoint \textsc{Round-Robin}
is $q=\left(k+1\right)r$, where~$r$ is the smallest real root of
$x=\left(1-x\right)^{k-1}$.
\end{proposition}
\begin{proof}
Denote by $t=\left(t_{i}\right)_{i=1}^{\infty}$ the sequence of
update times. A snapshot of \textsc{Round-Robin} is $S_{i+k}=\left(t_{i+1},t_{i+2},\ldots,t_{i+k}\right)$. For the
algorithm to be $\lambda\left(k+1\right)$-efficient for some $\lambda\ge1$,
update times need to satisfy $t_{i+1}\le\lambda t_{i+k}$ as well
as the subgeometric conditions
\(
t_{i+k}\le t_{i+k-1}/\left(1-\lambda\right)\le\cdots\le\left(1-\lambda\right)^{1-k}t_{i+1},
\)
implying that $\lambda\le\left(1-\lambda\right)^{1-k}$. This is possible
if and only if $\lambda\ge r$. Moreover, for any $G>1$, the efficiency
of \textsc{Round-Robin} using a geometric update times sequence $\left( t_{i}=G^{i}\right) _{i=1}^{\infty}$
is $\max\left\{ 1-1/G,G^{1-k}\right\} \left(k+1\right)$. This is
indeed minimized by choosing $G=1/\left(1-r\right)$.
\end{proof}
\begin{remark*}
\textsc{Round-Robin} is pretty bad for large $k$; 
indeed, $\left(k+1\right)r\approx\ln k-\ln\ln k$ is asymptotically inferior
to the simple bound $q_{k}\le2$ from the introduction.
\end{remark*}
The case $k=2$ is made obvious by Property~\ref{property:never-update-k},
since without loss of generality \textsc{Round-Robin} is the \emph{only} 2-checkpoint algorithm to consider.
Thus $q_2=1.5$.
\begin{proposition}
For $k=3$ we have $q_{3}=4r_{3}\approx1.52786$, where $r_{3}=\frac{3-\sqrt{5}}{2}\approx0.38197$
is the smaller root of $x^{2}-3x+1=0$.
\end{proposition}
\begin{proof}
For the upper bound, \textsc{Round-Robin} is $4r_3$-efficient.
Note that $G_3:=1/\left(1-r_3\right)=\frac{1+\sqrt{5}}{2}$ is the golden ratio.

For the lower bound, consider a $4r$-efficient $3$-checkpoint algorithm and 
a snapshot $S_{i}=\left(x,y,t_i\right)$.
By subgeometry we must have $t_i\le y/\left(1-r\right)\le x/\left(1-r\right)^2$
and for subinterval 1 to be compliant we need $x\le rt_i$, which
together imply $\left(1-r\right)^2\le r$, i.e., $r^2-3r+1\le0$.
Thus $r\ge r_3$.
\end{proof}

\textsc{Round-Robin} is no longer optimal for $k>3$. Indeed, cyclic
algorithms with better efficiency were described in~\cite[Figure~3]{BDNS13}
for $k=4,5,6,7,8$. These provide upper bounds on $q_{4},\ldots,q_{8}$,
respectively. Nevertheless, no formal proof of optimality was provided.
\begin{remark*}
These algorithms were found by the use of linear programming, which
is thoroughly discussed in Section~\ref{sec:LP}.
\end{remark*}

For $k=4,5$ the optimal algorithms are $2$-cyclic; $k=5$ is geometric while $k=4$ is not.

\begin{proposition}
\label{prop:k5}For $k=5$ we have $q_{5}=6r_{5}\approx1.47073$,
where $r_{5}\approx0.24512$ is the (only) real root of $x^{3}-4x^{2}+5x-1=0$.
\end{proposition}
\begin{proof}
 For the upper bound, set $G=1/\left(1-r_{5}\right)$ and consider
the $\left(2,G\right)$-geometric algorithm with pattern $P=\left(3,1\right)$
and initial snapshot $S_{5}=\left(1,G^{2},G^{3},G^{4},G^{5}\right)$.
The algorithm is indeed $\left(2,G^{2}\right)$-cyclic, as $S_{6}=\left(1,G^{2},G^{4},G^{5},G^{6}\right)$
and $S_{7}=\left(G^{2},G^{4},G^{5},G^{6},G^{7}\right)=G^{2}S_{5}$.
Due to being geometric, we just need to verify compliance of subinterval
3 in $S_{5}$: $G^{4}-G^{2}\le r_{5}G^{6}$ and subinterval 1 in $S_{6}$:
$G^{2}\le r_{5}G^{7}$. Altogether it remains to show that
\[
r_{5}\ge\max\left\{ G^{-5},G^{-2}-G^{-4}\right\} =\max\left\{ \left(1-r_{5}\right)^{5},\left(1-r_{5}\right)^{2}-\left(1-r_{5}\right)^{4}\right\} ,
\]
both of which are equal to $r_{5}$.

For the lower bound, consider a $6r$-efficient $5$-checkpoint algorithm.
It cannot be \textsc{Round-Robin}, which satisfies $\left(1-r\right)^{4}\le r$
and in particular $r>0.275>r_{5}$. We thus pick a snapshot $S_{i}=\left(x,y,z,w,t_{i}\right)$
followed by update steps $p_{i+1}=1$ and $p_{i+2}\in\left\{ 2,3,4\right\} $.
Thus $S_{i+2}$ is one of $\left(y,w,t_{i},t_{i+1},t_{i+2}\right)$,
$\left(y,z,t_{i},t_{i+1},t_{i+2}\right)$, or $\left(y,z,w,t_{i+1},t_{i+2}\right)$.
In either case, compliance of $S_{i}$, $S_{i+1}$ and $S_{i+2}$
implies $t_{i}\le r\left(t_{i}+t_{i+1}+t_{i+2}\right)$, which together
with $t_{i}\le\left(1-r\right)t_{i+1}\le\left(1-r\right)^{2}t_{i+2}$
yields $(\left(1-r\right)^{2}-r)\left(1-r\right)\le r$,
i.e., $r^{3}-4r^{2}+5r-1\ge0$. Thus $r\ge r_{5}$.
\end{proof}

\begin{proposition}
\label{prop:k4}For $k=4$ we have $q_{4}=5r_{4}\approx1.53989$,
where $r_{4}=\left(2+2\cos\left(2\pi/7\right)\right)^{-1}\approx0.307979$
is the smallest root of $x^{3}-5x^{2}+6x-1=0$. Moreover, the efficiency
of any geometric 4-checkpoint algorithm is at least $5\tilde{r}_{4}\approx1.58836$,
where $\tilde{r}_{4}\approx0.31767>r_{4}$ is the real root of $x^{3}-3x^{2}+4x-1=0$.
\end{proposition}
\begin{proof}
 For the upper bound, let $\gamma=1/\sqrt{r_{4}}$ be the largest
root of $x^{3}-x^{2}-2x+1=0$ and consider the $\left(2,\gamma\right)$-cyclic
algorithm with pattern $P=\left(3,1\right)$ starting at $S_{4}=\left(1,\gamma,\gamma^{3}-\gamma^{2},\gamma^{2}\right)$.
The algorithm is indeed $\left(2,\gamma\right)$-cyclic, as $S_{5}=\left(1,\gamma,\gamma^{2},\gamma^{4}-\gamma^{3}\right)$
and $S_{6}=\left(\gamma,\gamma^{2},\gamma^{4}-\gamma^{3},\gamma^{3}\right)=\gamma S_{4}$.
The update times sequence is $1/\left(1-r_{4}\right)$-subgeometric
if $\left(1-r_{4}\right)\gamma^{3}\le\gamma^{4}-\gamma^{3}$ and $\left(1-r_{4}\right)\left(\gamma^{3}-\gamma^{2}\right)\le\gamma^{2}$;
furthermore we need to verify compliance of subinterval 3 in $S_{5}$:
$\gamma^{2}-\gamma\le r_{4}\left(\gamma^{4}-\gamma^{3}\right)$ and
subinterval 1 in $S_{6}$: $\gamma\le r_{4}\gamma^{3}$. Altogether
it remains to show
\[
r_{4}\ge\max\left\{ 2-\gamma,1-\frac{1}{\gamma^{2}-\gamma},\frac{\gamma^{2}-\gamma}{\gamma^{4}-\gamma^{3}},\gamma^{-2}\right\} .
\]
Indeed the last three are equal to $r_{4}$ while the first is
\[
2-\gamma=\gamma^{-1}\left(1-\left(\gamma-1\right)^{2}\right)<\gamma^{-1}\left(1-\left(\gamma-1\right)\gamma^{-1}\right)=\gamma^{-2}=r_{4},
\]
using $\gamma-1>\gamma^{-1}$ since $r_{4}+\sqrt{r_{4}}<r_{3}+\sqrt{r_{3}}=1$.

For the lower bound, consider a $5r$-efficient 4-checkpoint algorithm.
If it is \textsc{Round-Robin}, it satisfies $\left(1-r\right)^{3}\le r$
so $r\ge\tilde{r}_{4}$. Otherwise pick a snapshot $S_{i}=\left(x,y,z,w\right)$
followed by update steps $p_{i+1}=1$ and $p_{i+2}\in\left\{ 2,3\right\} $.
The snapshot $S_{i+2}$ is $\left(y,w,t_{i+1},t_{i+2}\right)$ or
$\left(y,z,t_{i+1},t_{i+2}\right)$. The snapshot $S_{i+3}$ after
the next update step $p_{i+3}$ can be one of four:
\begin{itemize}
\item If $S_{i+3}=\left(z,t_{i+1},t_{i+2},t_{i+3}\right)$ then $t_{i+1}=\left(t_{i+1}-z\right)+z\le r\left(t_{i+2}+t_{i+3}\right)$
by compliance of $S_{i+2}=\left(y,z,t_{i+1},t_{i+2}\right)$ and $S_{i+3}$,
yielding $\left(1-r\right)^{3}\le r^{2}$, i.e., $r>0.43016>\tilde{r}_{4}>r_{4}$;
\item If $S_{i+3}$ is $\left(y,t_{i+1},t_{i+2},t_{i+3}\right)$ or $\left(w,t_{i+1},t_{i+2},t_{i+3}\right)$
then $t_{i+1}=y+\left(t_{i+1}-y\right)\le r\left(t_{i+1}+t_{i+3}\right)$
by compliance of $S_{i+1}$ and $S_{i+3}$, or $w\le rt_{i+3}$ by
compliance of $S_{i+3}$; either way yields $\left(1-r\right)^{3}\le r$,
which again implies $r\ge\tilde{r}_{4}$;
\item The only remaining option is $S_{i+3}=\left(y,w,t_{i+2},t_{i+3}\right)$,
which means $p_{i+2}=3$, i.e., $S_{i+2}=\left(y,w,t_{i+1},t_{i+2}\right)$.
Now
\begin{gather*}
\left(1-r\right)t_{i+1}\le w=y+\left(w-y\right)\le r\left(t_{i+1}+t_{i+2}\right)\\
t_{i+2}=y+\left(w-y\right)+\left(t_{i+2}-w\right)\le r\left(t_{i+1}+t_{i+2}+t_{i+3}\right)
\end{gather*}
by compliance of all four snapshots, so
\begin{gather*}
\left(1-2r\right)t_{i+1}\le rt_{i+2}\\
\left(1-r\right)^{2}t_{i+2}\le\left(1-r\right)r\left(t_{i+1}+t_{i+3}\right)\le\left(1-r\right)rt_{i+1}+rt_{i+2}
\end{gather*}
hence $\left(1-2r\right)\left(\left(1-r\right)^{2}-r\right)\le r^{2}\left(1-r\right)$,
i.e., $r^{3}-6r^{2}+5r-1\ge0$, which implies $r\ge r_{4}$.
\end{itemize}
Note that a $1/\left(1-r\right)$-geometric update times sequence
would satisfy $\left(1-r\right)t_{i+2}=t_{i+1}$ in the last case,
yielding $\left(1-r\right)^{3}\le r$ one last time and implying $r\ge\tilde{r}_{4}$.
\end{proof}


\subsection{\label{sec:LP}Casting the problem as a linear program}

Fix $\lambda\ge1$ and an update pattern $p=\left(p_{i}\right)_{i=k+1}^{\infty}$.
Can we choose a sequence $t=\left(t_{i}\right)_{i=1}^{\infty}$ of
update times such that the resulting $k$-checkpoint algorithm $\textsc{Alg}=\left(t,p\right)$
is $\lambda$-efficient?

Each snapshot $S_{i}$ consists of a particular subset of the variables
$t$, and using $p$ we can determine exactly which. Furthermore,
all constraints (e.g., monotonicity, subgeometry, compliance) can
be expressed as linear inequalities. This gives rise to an infinite
linear program $L=L\left(\lambda;p\right)$, which is feasible whenever
a $\lambda$-efficient algorithm with the prescribed pattern $p$
exists. Note that all constraints are homogeneous, so to avoid the
zero solution we add the non-homogeneous condition $t_{k}=1$.

In addition, we are not interested in solutions where $t$ is bounded.
This can happen, for instance, when $p_{i}=k-1$ for all $i\ge k+1$.\footnote{%
This may not seem a valid pattern to consider, given Property~\ref{property:update-1-infinitely-often};
however, when solving a finite subprogram we might have to consider
an arbitrarily long prefix of the pattern with no occurrences of 1.} Luckily, by using Property~\ref{property:supergeometric} we can restrict
our attention to exponentially increasing sequences $t$, so we add to $L$
the linear inequalities from Properties~\ref{property:supergeometric} and~\ref{property:supergeometric2}.
Now $L$ is feasible if and only if a $\lambda$-efficient algorithm with the prescribed pattern $p$ exists;
in other words, $q_k$ is the infimum\footnote{%
This infimum is actually a minimum, by~\cite[Theorem~8]{BDNS13} and also by Proposition~\ref{prop:LP-discrete-regimes}.} over $\lambda\ge1$
for which there exists a pattern~$p$ such that $L\left(\lambda;p\right)$
is feasible.

As an infinite program, $L$ is not too convenient to work with. We
can thus limit our attention to finite subprograms $L\left(\lambda;\left(p_{k+1}, \ldots, p_{k+n}\right)\right)$
for some $n\in\mathbb{N}$, which only involve the $k+n$ variables
$t_1, \ldots,t_{k+n}$ and the relevant $3k+6n$ constraints.
Finite subprograms can no longer ensure the existence of a $\lambda$-efficient
algorithm, but can be used to prove lower bounds on $q_{k}$ in the
following way. Write $\Sigma=\left\{ 1,\ldots,k-1\right\} $ and consider
the set $\Sigma^{*}$ of strings, i.e, finite sequences over $\Sigma$.
\begin{defn*}
A string $B\in\Sigma^{*}$ is called a {\em $\lambda$-witness} if $L\left(\lambda;B\right)$
is infeasible.
A string set $\mathcal{B}\subset\Sigma^{*}$ is called {\em blocking} if
any infinite sequence $p$ over $\Sigma$ contains some $B\in\mathcal{B}$
as a substring.
\end{defn*}
\begin{property_XXX}\label{property:LB-via-LP}
If there exists a blocking set of $\lambda$-witnesses
for some $\lambda\ge1$, then $q_{k}>\lambda$.
\end{property_XXX}
\begin{remark*}
The lower bound of Property~\ref{property:LB-via-LP} holds for all algorithms, cyclic or not.
\end{remark*}

We now describe a strategy to approximate $q_{k}$ to arbitrary precision.
For the lower bound we use Property~\ref{property:LB-via-LP}; for the upper bound, we limit our
focus to cyclic algorithms. Given $\gamma>1$ and a string $P\in\Sigma^{*}$
of length $n$, we can augment $L\left(\lambda;P\right)$ with $k$
equality constraints $\left\{ t_{i+n}=\gamma\cdot t_{i}\right\} _{i=1}^{k}$;
call the resulting program $L^{*}\left(\lambda,\gamma;P\right)$.
This is a finite linear program, which we can computationally solve
given $\lambda$, $\gamma$, and $P$. Although $\gamma\le G^{n}$
is not known to us, we can first compute an approximation~$\tilde{\gamma}$
of~$\gamma$ by solving $L_{10n}\left(\gamma;P^{10}\right)$, and
then solve $L^{*}\left(\lambda,\tilde{\gamma};P\right)$. Using binary
search, we can compute a numerical approximation $\tilde{\lambda}$
of the minimal $\lambda$ for which $L^{*}\left(\lambda,\tilde{\gamma};P\right)$
is feasible. Lastly, we can enumerate short strings $P\in\Sigma^{*}$
in a BFS/DFS-esque manner and take the best $\tilde{\lambda}$ obtained.

To demonstrate this strategy, we computed $q_{2},\ldots,q_{10}$ up
to 7 decimal digits, using a Python program employing GLPK~\cite{GLPK} via CVXOPT~\cite{CVXOPT} 
(see Table~\ref{tbl:comput-UB}; starred values of $k$ are geometric algorithms).
\begin{table}[h]
\centering
\begin{tabular}{|c|c|c|l|c|c|c|}
\hline
$k$  & $q_k$  & $\gamma$  & \multicolumn{1}{c|}{$P$ }  & $n$  & $\left|\mathcal{B}\right|$  & $\displaystyle\max_{B\in\mathcal{B}}\left|B\right|$\tabularnewline
\hline
\hline
2{*}  & 1.5  & 2 & $(1)$  & 1  & 0  & 0\tabularnewline\hline
3{*}  & 1.5278641  & 1.618037  & (1)  & 1  & 2  & 1\tabularnewline\hline
4  & 1.5398927  & 1.8019377  & (1,3)  & 2  & 7  & 3\tabularnewline\hline
5{*}  & 1.4707341  & 1.7548777  & (1,3)  & 2  & 36  & 13\tabularnewline\hline
6  & 1.5127400 & 3.627365  & (1,2,3,1,3,5)  & 6  & 117  & 9\tabularnewline\hline
7  & 1.4974818  & 3.11201  & (1,3,4,1,5,3)  & 6  & 559  & 10\tabularnewline\hline
8  & 1.4851548  & 10.712656  & (1,2,4,7,5,3,1,7,5,3,7,1,4,2,4,5)  & 16  & 1698  & 14\tabularnewline\hline
9  & 1.4730721  & 3.2748095  & (1,5,3,5,1,5,6,3)  & 8  & 5892  & 135\tabularnewline\hline
10  & 1.4678452  & 5.67943  & (1,5,3,5,1,5,6,3,1,5,9,3,5,9)  & 14  & 32843  & 20\tabularnewline\hline
\end{tabular}
\caption{\label{tbl:comput-UB}Computationally-verified bounds on $q_{k}$
for $2\le k\le10$.}
\end{table}

At first it seems that Property~\ref{property:LB-via-LP} cannot be used
to pinpoint $q_{k}$ exactly, since any finite blocking set $\mathcal{B}$
of $\left(q_k-\epsilon\right)$-witnesses for some $\epsilon>0$
leaves an interval of uncertainty of length~$\epsilon$.
The following proposition eliminates this uncertainly.
\begin{proposition}
\label{prop:LP-discrete-regimes}For every string $B\in\Sigma^{*}$
there is a finite set $\Lambda_{B}\subset\mathbb{R}$ such that the
feasibility of $L\left(\lambda;B\right)$ for some $\lambda\ge1$
only depends on the relative order between $\lambda$ and members
of~$\Lambda_{B}$. In particular, there exists some~$\epsilon>0$
such that if $L\left(\lambda;B\right)$ is feasible and~$B$ is a
$\left(\lambda-\epsilon\right)$-witness, then~$B$ is also a $\lambda'$-witness
for all $\lambda-\epsilon<\lambda'<\lambda$.
\end{proposition}

\begin{proof}
Fix $B\in\Sigma^{*}$. Treating $\lambda$ as a parameter, note that
the subprogram $L\left(\lambda;B\right)$ is feasible if and only
if its feasible region, the convex polytope $\mathcal{P}\left(\lambda;B\right)$,
is nonempty. Decreasing $\lambda$ shrinks $\mathcal{P}\left(\lambda;B\right)$
until some critical $\lambda_{B}$ for which $\mathcal{P}\left(\lambda_{B};B\right)$
is reduced to a single vertex, at which a subset of the linear constraints
are satisfied with equality. Hence $\lambda_{B}$ is a solution of
some polynomial equation determined by the relevant constraints. The
set of constraints is finite, thus there are finitely many polynomial
equations that can define $\lambda_{B}$, and we can take $\Lambda_{B}$
as the set of all their roots. Now take $\epsilon$ to be smaller
than the distance between any two distinct elements of $\Lambda_{B}$.
\end{proof}
\begin{remark*}
Note that when $\epsilon$ is small enough, we can actually retrieve
the polynomial equations defining $\lambda$ and $\gamma$ from the
polytope $\mathcal{P}^{*}\big(\tilde{\lambda},\tilde{\gamma};B\big)$;
using this method we get an algebraic representation of $q_{k}$ rather
than a rational approximation. To demonstrate, $q_{9}=10r_{9}$, where
$r_{9}\approx0.1473072131$ is the smallest real root of
\[
x^{8}-7x^{7}+22x^{6}-40x^{5}+39x^{4}-17x^{3}+10x^{2}-8x+1.
\]
\end{remark*}


\section{\label{sec:recursive-construction}Asymptotically Optimal Upper Bounds}

In this section we describe a family of geometric $k$-checkpoint algorithms.
Despite our experience from Table~\ref{tbl:comput-UB}---that only for $k=2,3,5$ optimal algorithms are geometric---this family is rich enough to be asymptotically optimal,
i.e., $\left(1+o\left(1\right)\right)q_{k}$-efficient.


\subsection{A recursive geometric algorithm}

Fix a real number $G>1$ and an integer $m\ge0$. We describe a $k$-checkpoint
algorithm \textsc{Recursive$\left(G,K\right)$}, where $K$ is an
$\left(m+2\right)$-subset $\left\{ 0,\ldots,k\right\} $ whose elements
are
\[
k=k_{0}>k_{1}>k_{2}>\cdots>k_{m}>k_{m+1}=0.
\]
\textsc{Recursive$\left(G,K\right)$} is $\left(2^{m},G\right)$-geometric,
and its update pattern $p$ is defined as $p_{k+i}=1+k_{\mu\left(i\right)+1}$,
where $\mu\left(i\right)$ is the largest $\mu\le m$ for which $2^{\mu}$
divides $i$. It is easy to see that $p$ is $2^{m}$-periodic, and
we can just refer to $P=\left(p_{k+i}\right)_{i=1}^{2^{m}}$. As per
Remark~\ref{rem:rebase}, via rebasing there is no need to explicitly define
the initial snapshot $S_{k}$.
\begin{example*}
For $K=\left\{ 0,2,4,9,19\right\} $ we get $P=(10,5,10,3,10,5,10,1)$.
\end{example*}
True to its name, \textsc{Recursive$\left(G,K\right)$} can be viewed
also as a recursive algorithm: the base case $m=0$ (i.e., $K=\left\{ 0,k\right\} $)
is simply $k$-checkpoint \textsc{Round-Robin}; for $m\ge1$, \textsc{Recursive$\left(G,K\right)$}
alternates between updating the $\left(k_{1}+1\right)$-st oldest
checkpoint and between acting according to the inner $k_{1}$-checkpoint
algorithm \textsc{Recursive$\left(G^{2},K\setminus\left\{ k\right\} \right)$}.

Let us elaborate a bit more on the recursive step. In every snapshot
$S_{i}=\left(T_{1},\ldots,T_{k}\right)$ we have $T_{j}=G^{i+j}$
for $k_{1}+1\le j\le k$ since we never update checkpoints younger
than $k_{1}+1$. In every \emph{odd} snapshot $S_{i}$ we have just
updated the $\left(k_{1}+1\right)$-st oldest checkpoint, so $T_{k_{1}}=G^{i+k_{1}-1}$
while $T_{k_{1}+1}=G^{i+k_{1}+1}$. This means that $\log_{G}T_{j}$
for $j=1,\ldots,k_{1}$ all have the same parity as $i+k_{1}-1$ in
any snapshot $S_{i}$. We thus treat $S'=\left(T_{1},\ldots,T_{k_{1}}\right)$
as a snapshot of a $k_{1}$-checkpoint algorithm, which operates at
half speed and never sees half of the checkpoints. The inner algorithm
can rightfully be called \textsc{Recursive$\left(G^{2},K\setminus\left\{ k\right\} \right)$},
as the common ratio of the update times sequence for the checkpoints
that do make it to the inner algorithm is $G^{2}$, and taking only
the even locations of $P$ yields a $2^{m-1}$ periodic sequence $p'$
such that $p_{k+i}'=p_{k+2i}=1+k_{\mu\left(2i\right)+1}=1+k_{\mu\left(i\right)+2}$.


\subsection{Analyzing the recursive algorithm}

First we determine exactly how efficient \textsc{Recursive$\left(G,K\right)$}
can be for any $G$ and $K$, and then we work with a particular choice.

Denote by $r\left(G,K\right)$ the maximum of%
\begin{subequations}
\begin{gather}
1-G^{-1};\label{eq:cond1}\\
\max\left\{ G^{-e\left(\ell\right)}\left(G^{2^{\ell}}-G^{-2^{\ell}}\right)\right\} _{\ell=0}^{m-1};\mbox{ and}\label{eq:cond2}\\
G^{2^{m}-e\left(m\right)},\label{eq:cond3}
\end{gather}
\end{subequations}%
where
\[
e\left(\ell\right)=\sum_{j=0}^{\ell}2^{j}\left(k_{j}-k_{j+1}\right)\text{ for }\ell=0,1,\ldots,m.
\]
\begin{theorem}
\label{thm:recursive-bound}Given $G$ and $K$, the efficiency of
\textsc{Recursive$\left(G,K\right)$} is $(k+1)r\left(G,K\right)$.
\end{theorem}
\begin{proof}
Write $r=r\left(G,K\right)$ and denote by $q=q\left(G,K\right)$
the efficiency of \textsc{Recursive$\left(G,K\right)$}. Our aim is
to show that $q=\left(k+1\right)r$; namely, that \textsc{Recursive$\left(G,K\right)$}
is $\left(k+1\right)r$-efficient, but is not $\tilde{q}$-efficient
for any $\tilde{q}<\left(k+1\right)r$.

We proceed by induction on~$m$. The base case $m=0$ is \textsc{Round-Robin}, whose efficiency
by the proof of Proposition~\ref{prop:RR} is $q=\left(k+1\right)\cdot\max\left(1-G^{-1},G^{1-k}\right)=\left(k+1\right)r$.
Assume now $m\ge1$ and consider a critical subinterval of length
$\delta=qT_{k}/\left(k+1\right)$ in a snapshot $S_{i}=\left(T_{1},\ldots,T_{k}\right)$,
taken at time $T_{k}=G^{i+k}$. This subinterval can be of one of three
types:
\begin{enumerate}
\item The last one, between $T_{k-1}$ and $T_{k}$. Here, by~\eqref{eq:cond1},
\[
\delta=T_{k}-T_{k-1}=\left(1-G^{-1}\right)T_{k}\le rT_k;
\]

\item One created by merging two smaller subintervals in an odd\footnote{When $i$ is odd (even) we call the snapshot $S_{i}$ odd
(even).} snapshot, between $T_{k_{1}}=G^{i+k_{1}-1}$ and $T_{k_{1}+1}=G^{i+k_{1}+1}$.
Here, by the case $\ell=0$ of \eqref{eq:cond2} 
\[
\delta=T_{k_{1}+1}-T_{k_{1}}=\left(G^{1}-G^{-1}\right)G^{i+k_{1}}=G^{k_{1}-k}\left(G^{1}-G^{-1}\right)T_{k}\le rT_k;
\]

\item One created by merging two smaller subintervals in an even
snapshot. This subinterval is among $S'=\left(T_{1},\ldots,T_{k_{1}}\right)$,
which we treat as a snapshot of \textsc{Recursive$\left(G^{2},K\setminus\left\{ k\right\} \right)$,}
taken at time $T_{k_{1}}=G^{i+k_{1}}$. Denote the efficiency of
\textsc{Recursive$\left(G^{2},K\setminus\left\{ k\right\} \right)$} by $q'=\left(k_{1}+1\right)r'$;
thus 
\[
\delta\le r'T_{k_{1}}=r'G^{i+k_{1}}=r'G^{k_{1}-k+i+k}=G^{e\left(0\right)}r'T_{k},
\]
with equality whenever this subinterval is critical for \textsc{Recursive$\left(G^{2},K\setminus\left\{ k\right\} \right)$}.
By the induction hypothesis, $r'=r\left(G^{2},K\setminus\left\{ k_{0}\right\} \right)$
is the maximum of 
\begin{gather*}
1-\left(G^{2}\right)^{-1};\\
\max\left\{ \left(G^{2}\right)^{-e'\left(\ell\right)}\left(\left(G^{2}\right)^{2^{\ell}}-\left(G^{2}\right)^{-2^{\ell}}\right)\right\} _{\ell=0}^{m-2};\mbox{ and}\\
\left(G^{2}\right)^{2^{m-1}-e'\left(m-1\right)},
\end{gather*}
where 
\[
e'\left(\ell\right)=\sum_{j=0}^{\ell}2^{j}\left(k_{j+1}-k_{j+2}\right)=\frac{1}{2}\left(e\left(\ell+1\right)-e\left(0\right)\right).
\]
To show $q\le\left(k+1\right)r$, it thus remains to show that $r\ge G^{e\left(0\right)}r'$.
\begin{enumerate}[(\itshape i)]
\item If $r'=1-G^{-2}<G\left(1-G^{-2}\right)=G^{1}-G^{-1}$ then $r>G^{e\left(0\right)}r'$
by case $\ell=0$ of~\eqref{eq:cond2};
\item If
\[
r'=\max\left\{ G^{-2e'\left(\ell\right)}\left(G^{2^{\ell+1}}-G^{-2^{\ell+1}}\right)\right\} _{\ell=0}^{m-2}
\]
then $r=G^{e\left(0\right)}r'$ by the remaining cases of~\eqref{eq:cond2};
\item If $r'=G^{2^{m}-2e'\left(m-1\right)}$ then $r=G^{e\left(0\right)}r'$
by~\eqref{eq:cond3}.
\end{enumerate}
\end{enumerate}
Now assume \textsc{Recursive$\left(G,K\right)$} is $\tilde{q}$-efficient for $\tilde{q}=\left(k+1\right)\tilde{r}$, where $\tilde{r}<r$, i.e., $\tilde{r}$ is strictly smaller than~\eqref{eq:cond1},~\eqref{eq:cond2} and~\eqref{eq:cond3}. The critical subinterval of length $\delta$ cannot be of type~1 or~2, due to~\eqref{eq:cond1} and case $\ell=0$ of~\eqref{eq:cond2}, respectively, so it must be of type 3. But then we get strict inequalities for the inner algorithm, so actually $r'<r\left(G^{2},K\setminus\left\{ k_{0}\right\} \right)$, contradicting the induction hypothesis.
\end{proof}

Given an integer $k\ge2$, let $m=\left\lfloor \log_{2}k\right\rfloor -1$.
Define $K^{*}=\left\{ k_{0},\ldots,k_{m+1}\right\} $ by $k_{j}=\left\lfloor 2^{-j}k\right\rfloor $
for $j=0,\ldots,m$ and $k_{m+1}=0$. Note that $k_{0}=k$ and that
$k_{m}\in\left\{ 2,3\right\} $.
\begin{theorem}
\label{thm:asymptotic-upper-bound}\textsc{Recursive$\left(G,K^{*}\right)$}
is $q$-efficient for large enough $k$, where $G=e^{q/\left(k+1\right)}$ and
\[
q=\left(1+\frac{3}{\log_{2}k}\right)\frac{k+1}{k}\ln4=\left(1+o\left(1\right)\right)\ln4.
\]
\end{theorem}

\begin{proof}
By Theorem~\ref{thm:recursive-bound}, it suffices to verify that
$(k+1)r\left(G,K^{*}\right)<q$, that is, $r\left(G,K^{*}\right)\le\ln G$
for sufficiently large $k$. Clearly~\eqref{eq:cond1} holds since
$1-G^{-1}<\ln G$ for all $G>1$. It remains to verify~\eqref{eq:cond2}
and~\eqref{eq:cond3}, handled by Propositions~\ref{prop:gamma-bound-case2}
and~\ref{prop:gamma-bound-case3} respectively.
\end{proof}
\begin{proposition}
\label{prop:gamma-bound-case2}For $k\ge2^{13}$ and $G$ as above,
$G^{-e\left(\ell\right)}\left(G^{2^{\ell}}-G^{-2^{\ell}}\right)<\ln G$
for all $\ell=0,\ldots,m-1$.
\end{proposition}

\begin{proposition}
\label{prop:gamma-bound-case3}For $k\ge5$ and $G$ as above, $G^{2^{m}-e\left(m\right)}<\ln G$.
\end{proposition}

\begin{remark*}
Theorem~\ref{thm:asymptotic-upper-bound} chooses $G$ suboptimally.
Empirical evidence shows that, for all $k\ge2$, the optimal $G=G^{*}$
for \textsc{Recursive$\left(G,K^{*}\right)$} satisfies~\eqref{eq:cond1}
and one of~\eqref{eq:cond2} and~\eqref{eq:cond3}. In other words,
it is the smallest root of either $1-x^{-1}=x^{2^{m}-e\left(m\right)}$
or $1-x^{-1}=x^{-e\left(\ell\right)}\left(x^{2^{\ell}}-x^{-2^{\ell}}\right)$
for some $\ell=0,\ldots,m-1$
(see also Tables~\ref{tbl:recursive-best} and~\ref{tbl:recursive-worst}).
\end{remark*}
\begin{remark*}
With additional effort the constant $3$ in Theorem~\ref{thm:asymptotic-upper-bound}
can be improved by a factor of almost~6 to $\tau:=-\log_{2}\ln2\approx0.53$.
The major obstacle is that cases $\ell=m-2$ and $\ell=m-1$ of~\eqref{eq:cond2}
need to be done separately since the appropriate $f\left(x,z\right)$
in the proof of Proposition~\ref{prop:gamma-bound-case2} is negative
for $z<\log_{2}\left(\ln4/\left(1-\ln2\right)\right)\approx2.1756$.
No proof is possible for $\tau'<\tau$ since then~\eqref{eq:cond3}
would be violated for large enough $k=2^{m+2}-1$.
\end{remark*}
\begin{remark*}
We verified that the algorithm \textsc{Recursive$\left(G^{*},K^{*}\right)$} is
$\left(1+\tau/\log_{2}k\right)\frac{k+1}{k}\ln4$-efficient for $2\le k\le2^{13}$
as well. 
See Tables~\ref{tbl:recursive-best} and~\ref{tbl:recursive-worst} for $k=2^{m+1}$ (best case) and $k=2^{m+2}-1$ (worst case), respectively.
\end{remark*}

Before proving Propositions~\ref{prop:gamma-bound-case2} and~\ref{prop:gamma-bound-case3}
we would like to simplify $e\left(\ell\right)$ for our $K^{*}$.
\begin{claim*}
$e\left(\ell\right)>\left(\ell+1\right)k/2-2^{\ell}$ for $\ell=0,\ldots,m$;
moreover, $e\left(m\right)>\left(m+2\right)k/2-2^{m}$.
\end{claim*}

\begin{proof}
We have
\begin{align*}
e\left(\ell\right) & =\sum_{j=0}^{\ell}2^{i}\left(k_{j}-k_{j+1}\right)=k_{0}-2^{\ell}k_{\ell+1}+\sum_{j=1}^{\ell}2^{j-1}k_{i}\\
 & =k-2^{\ell}\left\lfloor 2^{-\ell-1}k\right\rfloor +\sum_{j=1}^{\ell}2^{j-1}\left\lfloor 2^{-j}k\right\rfloor \\
 & \ge k-\frac{k}{2}+\frac{1}{2}\sum_{j=1}^{\ell}\left(k-2^{j-1}\right)=\frac{k}{2}\left(\ell+1\right)+1-2^{\ell}.
\end{align*}
The slightly improved bound for $\ell=m$ is obtained by observing
that $2^{m}k_{m+1}=0$ in the above calculation since $k_{m+1}=0$.
\end{proof}
Write $x=\log_{2}k$ and observe that
$G^{k/2}=e^{\left(1+3/x\right)\ln2}=2^{1+3/x}$
and $G^{kx/2}=2^{x+3}=8k$ for our choice of $G$.

\begin{proof}[Proof of Proposition~\ref{prop:gamma-bound-case3}]
By the claim above
\[
e\left(m\right)-2^{m}>\left(m+2\right)k/2-2^{m+1}\ge kx/2-k
\]
so
\[
\ln G=\frac{q}{k+1}=\frac{1}{k}\ln\left(4^{1+3/x}\right)\ge\frac{4^{1+3/x}}{8k}=G^{k-kx/2}>G^{2^{m}-e\left(m\right)},
\]
where the first inequality is true since $\ln y\ge y/8$ for $2\le y\le26$
and indeed for $k\ge5$ we have $0<3/x<1.3$ and $4<4^{1+3/x}<25$.
\end{proof}
 
\begin{proof}[Proof of Proposition~\ref{prop:gamma-bound-case2}]
By the claim above
\begin{align*}
G^{-e\left(\ell\right)}\left(G^{2^{\ell}}-G^{-2^{\ell}}\right) & <G^{2^{\ell}-\left(\ell+1\right)k/2}\left(G^{2^{\ell}}-G^{-2^{\ell}}\right)\\
 & =2^{-\left(1+3/x\right)\left(\ell+1\right)}\left(G^{2^{\ell+1}}-1\right)\\
 & =2^{-\left(\ell+1\right)}\cdot8^{-\left(\ell+1\right)/x}\left(\exp\left(2^{\ell+1}\ln G\right)-1\right)\\
 & <8^{-\left(\ell+1\right)/x}\frac{\ln G}{1-2^{\ell+1}\ln G},
\end{align*}
using $e^{y}-1<y/\left(1-y\right)$ for all $0<y\le1$, and taking
$y=2^{\ell+1}\ln G$. Thus, it remains to show that
\[
8^{-\left(\ell+1\right)/x}\le1-2^{\ell+1}\ln G=1-2^{\ell+1-x}\left(1+3/x\right)\ln4.
\]
Let $z=x-\left(\ell+1\right)$ and note that $1\le z\le x-1$ as $0\le\ell\le m-1$.
Define
\[
f\left(x,z\right):=1-2^{1-z}\left(1+3/x\right)\ln2-8^{z/x-1};
\]
to conclude the proof we show that $f\left(x,z\right)$ is positive
for all $1\le z\le x-1$ and $x\ge13$.

First we compute some partial derivatives of $f$.
\begin{eqnarray*}
\frac{df\left(x,1\right)}{dx} & = & \frac{3\ln2}{x^{2}}\left(1+8^{1/x-1}\right);\\
\frac{df\left(x,x-1\right)}{dx} & = & \frac{3\ln2}{x^{2}}\left(2^{2-x}-8^{-1/x}\right)+2^{2-x}\left(1+\frac{3}{x}\right)\ln^{2}2;\\
\frac{\partial f\left(x,z\right)}{\partial z} & = & \left(1+\frac{3}{x}\right)2^{1-z}\ln^{2}2-\frac{3\ln2}{x}8^{z/x-1};\\
\frac{\partial^{2}f\left(x,z\right)}{\left(\partial z\right)^{2}} & = & -\left(1+\frac{3}{x}\right)2^{1-z}\ln^{3}2-\left(\frac{3\ln2}{x}\right)^{2}8^{z/x-1}.
\end{eqnarray*}
Now $\frac{df\left(x,1\right)}{dx}>0$ everywhere, so $f\left(x,1\right)$
is increasing and $f\left(x,1\right)\ge f\left(13,1\right)>0$ for
all $x\ge13$. Next $\frac{df\left(x,x-1\right)}{dx}<0$ for $x>7$,
so $f\left(x,x-1\right)$ is decreasing and
\[
f\left(x,x-1\right)>\lim_{x\to\infty}f\left(x,x-1\right)=\lim_{x\to\infty}1-8^{-1/x}=0
\]
for all $x>7$. Lastly, $f\left(x,z\right)$ is concave in $z$ as
$\frac{\partial^{2}f}{\left(\partial z\right)^{2}}<0$ everywhere.
Thus
\[
\min_{1\le z\le x-1}f\left(x,z\right)=\min\left\{ f\left(x,1\right),f\left(x,x-1\right)\right\} >0
\]
for all $x\ge13$.
\end{proof}


\section{\label{sec:lower-bound}Asymptotically Optimal Lower Bounds}

In this section we prove lower bounds on $q_{k}$, focusing on asymptotic
lower bounds in which~$k$ grows to infinity.

We start by reproving the simple asymptotic lower bound $q_{k}\ge2-\ln2-o\left(1\right)$ of~\cite[Theorem~6]{BDNS13},
and then improve it to $q_{k}\ge\ln4$, which is asymptotically optimal via the matching upper bound of Section~\ref{sec:recursive-construction}.


\subsection{Stability and bounding expressions}


Obtaining lower bounds requires viewing the problem from a different
perspective. It will sometimes be more convenient to refer to a certain
physical checkpoint, without considering its temporary freshness index $p$ in the checkpoint
sequence at some snapshot $S$ (which is variable and depends on $S$).

Given a $k$-checkpoint algorithm, we define a function $BE\left(s\right)$
and use it to bound its efficiency from below. The parameter $s$
is related to the notion of stability, which we now define.
\begin{defn*}
Fix a $k$-checkpoint algorithm. A checkpoint updated at time $T$ is called
$s$-stable, for some $s=1,\ldots,k-1$, if at least $s$ previous
checkpoints are updated before the next time it is updated.
\end{defn*}

By Property~\ref{property:update-1-infinitely-often} we can assume all
checkpoints get updated eventually; this means that in a snapshot $S=\left(T_{1},\ldots,T_{k}\right)$,
where $T_{k}$ is a time by which all checkpoints have been updated from
the initial snapshot, we have that the checkpoint updated at time $T_{k-s}$
is $s$-stable for $s=1,\ldots,k-1$.

For convenience, the proofs in this section assume the update times
sequence is normalized by a constant. This is captured by the following
definition.
\begin{defn*}
A $k$-checkpoint algorithm is called $s$-normalized if an $s$-stable checkpoint
is updated at time $R_{0}=1$.
\end{defn*}

Given an $s$-normalized $k$-checkpoint algorithm, we define a sequence
of times $1=R_{0}<R_{1}<\cdots<R_{s}$ as follows: $R_{i}$ for $i\ge1$
is the time at which the $i$-th checkpoint is \emph{removed} from $\left(0,1\right]$.
In other words, $R_{1}$ is the time $T>R_{0}$ at which some checkpoint
is updated; $R_{2}$ is the time $T>R_{1}$ at which we update the
next checkpoint that was previously updated in $\left(0,1\right]$ (but
not at $T>1$), and so forth. Note that the checkpoint updated at time
$R_{0}$ is not updated at any time $R_{i}$ for $i=1,\ldots,s$ by
the definition of stability. Now we are ready to define $BE\left(s\right)$.
\begin{defn*}
The $T$-truncated bounding expression of an $s$-normalized $k$-checkpoint algorithm
is $BE_{T}(s)=\sum_{i=1}^{s}U_{i}$, where $U_{i}=\min\left\{T,R_{i}\right\} $.
\end{defn*}

The bounding expression plays a crucial role in proving lower bounds,
based on Proposition~\ref{prop:BE-at-least-1} below. We note that
the truncated bounding expression only depends on the algorithm's behavior
until time $T$, and hence the bounds that can be obtained from it
are not tight for $k>3$. Nevertheless, the lower bound we obtain
using $BE_{2}$ in Corollary~\ref{cor:LB-optimal} is asymptotically
optimal, since the gap between it and the upper bound of Theorem~\ref{thm:asymptotic-upper-bound}
tends to zero as $k$ grows to infinity.
\begin{remark*}
It is possible to analyze $BE_{T}$ beyond $T=2$ and obtain tight
lower bounds for larger values of $k$. However, there is no asymptotic
improvement and the analysis becomes increasingly more technical as
$k$ grows.
\end{remark*}


\subsection{\label{sec:LB-weaker}Asymptotic lower bound of $2-\ln2\approx1.3068$}

To simplify the analysis, we assume $k$ is even. It can be extended
to cover odd values of $k$ as well, but this gives no asymptotic
improvement since $q_{k+1}\le q_{k}\cdot\frac{k+2}{k+1}$ for all $k$,
so we only lose an error term of $O\left(1/k\right)$, which is of
the same order as the error terms in Corollaries~\ref{cor:LB-suboptimal}
and~\ref{cor:LB-optimal}.

To simplify our notation we write $b=q/(k+1)$ throughout this section.
\begin{proposition}
\label{prop:BE-at-least-1}Any $\left(k/2\right)$-normalized $b(k+1)$-efficient
$k$-checkpoint algorithm satisfies $BE_{2}(k/2)\geq1/b$.
\end{proposition}

\begin{proof}
At time $R_{0}=1$, the time interval $(0,1]$ contains $k$ subintervals
of length $\le b$, giving rise to the inequality $b\cdot k\geq1$.
At time $R_{1}$, a checkpoint is removed from $(0,1]$ and it now contains
one subinterval of length $\le b\cdot R_{1}$ (two previous subintervals,
each of length $\le b$, were merged), and $k-2$ subintervals of length
$\le b$, giving $b\cdot(k-2+R_{1})\geq1$.

At time $R_{2}$, an additional checkpoint is removed from the time interval
$(0,1]$, hence it must contain an subinterval of length $\le b\cdot R_{2}$
formed by merging two previous subintervals. We obtain $b\cdot(k-4+R_{1}+R_{2})\geq1$,
since the remaining $k-3$ subintervals must include $k-4$ subintervals
of length $\le b$ and one (additional) subinterval of length at most
$\le b\cdot R_{1}$. Note that this claim holds regardless of which
checkpoint is updated at $R_{2}$, and it holds in particular in case
one of the subintervals merged at time $R_{2}$ contains the subintervals
merged at $R_{1}$ (in fact, this case gives the stronger inequality
$b\cdot(k-3+R_{2})\geq1$).

In general, for $j=1,2,\ldots,k/2$, at time $R_{j}$ the time interval
$(0,1]$ must contain $j$ distinct subintervals of lengths $\le b\cdot R_{i}$
for $i=1,2,\ldots,j$, and $k-2j$ subintervals of length $\le b$. This
gives the inequality
$k-2j+\sum_{i=1}^{j}R_{i}\geq1/b.$

Let $j\le k/2$ be the largest index such that $R_{j}\leq2$, so $U_{i}=R_{i}$
for all $1\le i\le j$ and $U_{i}=2$ for all $j<i\le k/2$. Now at
time $R_{j}$ we have
\[
1/b\leq k-2j+\sum_{i=1}^{j}R_{i}=\left(k/2-j\right)\cdot2+\sum_{i=1}^{j}U_{i}=\sum_{i=j+1}^{k/2}U_{i}+\sum_{i=1}^{j}U_{i}=BE_{2}(k/2).
\]
\end{proof}

Now we need an upper bound on the bounding expression. For the simpler
lower bound of $2-\ln2$ we use the following proposition.
\begin{proposition}
\label{prop:Ri-bounded}Any $\left(k/2\right)$-normalized $b(k+1)$-efficient
$k$-checkpoint algorithm satisfies $R_{i}\leq1/\left(1-bi\right)$ for $i=1,\ldots,k/2$.
\end{proposition}

\begin{proof}
At time $T=1/\left(1-bi\right)$, all subintervals are of length at most
$bT=b/\left(1-bi\right)$. Since $T-R_{0}=1/\left(1-bi\right)-1=bi/\left(1-bi\right)$,
for any $\epsilon>0$ the time interval $(R_{0},T+\epsilon]$ must consist
of at least $i+1$ subintervals, implying that the $i$-th checkpoint was
removed from the time interval $(0,1]$ by time $T$.
\end{proof}

\begin{proposition}
\label{prop:BE-weak-upper-bound}Let $b<\frac{1}{2}$. Any $\left(k/2\right)$-normalized
$b(k+1)$-efficient $k$-checkpoint algorithm satisfies
$b\cdot BE_{2}(k/2)\le\ln2+b/(1-2b)+bk-1.$
\end{proposition}

\begin{corollary}
\label{cor:LB-suboptimal}For all even $k\ge4$ we have $q_{k}\ge2-\ln2-o\left(1\right)$.
\end{corollary}

\begin{proof}
Fix a $\left(k/2\right)$-normalized $q_{k}$-efficient $k$-checkpoint algorithm, and let $b=q_{k}/(k+1)<\frac{1}{2}$. By Propositions~\ref{prop:BE-at-least-1}
and~\ref{prop:BE-weak-upper-bound} we have
 \[
 q_{k}=bk+b\ge2-\ln2-\frac{b}{1-2b}+b=2-\ln2-\frac{2b^2}{1-2b}\ge2-\ln2-\frac8{(k+1)(k-3)}.
 \]
\end{proof}

\begin{proof}[Proof of Proposition~\ref{prop:BE-weak-upper-bound}]
For $i\leq\lfloor1/2b\rfloor$, we have $R_{i}\leq1/\left(1-b\lfloor1/2b\rfloor\right)\leq2$,
but we cannot assure that $R_{i}\leq2$ for $i>\lfloor1/2b\rfloor$.
By Proposition~\ref{prop:Ri-bounded},
\[
b\cdot BE_{2}\left(k/2\right)\leq b\cdot\left(\sum_{i=1}^{\lfloor1/2b\rfloor}R_{i}+\sum_{i=\lfloor1/2b\rfloor+1}^{k/2}2\right)\leq\sum_{i=1}^{\lfloor1/2b\rfloor}\frac{b}{1-bi}+bk-1.
\]
Now for $b<\frac12$ the sum on the right-hand side can be bounded
by
\begin{align*}
\sum_{i=1}^{\lfloor1/2b\rfloor}\frac{1}{1/b-i} & \le\intop_{0}^{1/2b}\frac{dx}{1/b-x-1}=\ln\left(\frac{1}{b}-1\right)-\ln\left(\frac{1}{2b}-1\right)\\
 & =\ln2+\ln\left(1+\frac{b}{1-2b}\right)\le\ln2+\frac{b}{1-2b},
\end{align*}
establishing the proposition.
\end{proof}


\subsection{Improved asymptotic lower bound of $\ln4\approx1.3863$}

We now improve the asymptotic lower bound to $\ln4$. This result
is a simple corollary of the following lemma, which gives a tighter
upper bound on the bounding expression.
Recall that $q=b(k+1)$ and thus $G=(k+1)/(k+1-q)=1/(1-b)$.
\begin{lemma}
\label{lem:BE-upper-bound} For any $s$-normalized $b(k+1)$-efficient
$k$-checkpoint algorithm such that $1\le s\le k/2$ and $G^{k/2}\le2$
we have $b\cdot BE_{2}(s)\le G^s-1$.%
\end{lemma}
\begin{proof}
The proof is by induction on $s$. 
For the base case $s=1$, one checkpoint is updated at time at
most $G$, giving $b\cdot BE_{2}(1)\leq bG=G-1$.

Assume the hypothesis holds for all $i\leq s-1$ and our goal is to
prove it for $i=s>1$. Without loss of generality we assume the algorithm
satisfies all properties of Section~\ref{sec:basic}. Consider a
snapshot at time $T=2$, and denote by $R'$ the update time of the
first checkpoint in the time interval $(R_{0},T]=(1,2]$ at the snapshot
time $T=2$. We would like to apply the induction hypothesis from
time $R'$, but this cannot be done directly since it is not guaranteed
that the checkpoint last updated at $R'$ in the snapshot at time $T=2$
is $\left(s-1\right)$-stable (potentially, less than $s$ checkpoints
are removed from $(0,1]$ at $T=2$). To overcome this problem, recall
that the truncated bounding expression $BE_{2}$ only considers the
algorithm up to time $T=2$ by setting $U_{i}=\min\left\{ 2,R_{i}\right\}$.
Consequently, we can analyze a slightly different algorithm with the
same bounding expression $BE_{2}(s)$ in which the checkpoint at $R'$
in time $T=2$ is $\left(s-1\right)$-stable.\footnote{There are other ways to solve the problem and apply the induction
hypothesis, e.g., by extending the definition of a stable checkpoint.
However, this seems to require slightly more complex definitions and
induction hypothesis.} The modification is simple: if the original algorithm removes $s'\geq s$
checkpoint from $(0,1]$ in $(1,2]$, no change is required; otherwise,
$s'<s$ and the modified algorithm would simply remove $s-s'$ additional
arbitrary checkpoints from $(0,1)$ at time $T=2$. This transformation
leaves $BE_{2}(s)$ unchanged, and we can analyze it instead. Note
that the modified algorithm maintains all properties of Section~\ref{sec:basic}
at times $T<2$.

We first consider the case in which there is no checkpoint update in the
time interval $(R_{0}=1,R')$, implying that $R'=R_{1}\le G$.
We can now apply the induction hypothesis from $T=R_{1}$ with $i=s-1$
since at least $s-1$ checkpoint are removed from $(0,R_{1})$ before
the checkpoint at $R_{1}$ is updated again, namely, the checkpoint at $R_{1}$
is $\left(s-1\right)$-stable (we have an $\left(s-1\right)$-normalized
checkpoint algorithm). Therefore
\begin{align*}
b\cdot BE_{2}(s) & \leq b\cdot BE_{2}(s-1)\cdot R_{1}+b\cdot R_{1}=\left(b\cdot BE_{2}(s-1)+b\right)\cdot R_{1}\\
 & \leq\left((G^{s-1}-1)+b\right)\cdot G=G^s-1.
\end{align*}
Note that the multiplication of $BE_{2}(s-1)$ with $R_{1}$ undoes
the normalization of the bounding expression at time $T=R_{1}$, and
the addition with $b\cdot R_{1}$ is because $BE_{2}(s)$ should account
for $R_{1}$, but $BE_{2}(s-1)$ should not.

We also note that this actually proves a (slightly) stronger result,
since when calculating $BE_{2}(s)$ from $T=1$, we do not add terms
larger than $T=2$, but when calculating $BE_{2}(s-1)$ from $T=R_{1}$,
the restriction is looser, i.e., not adding terms larger than $T=2\cdot R_{1}>2$.
Therefore, if $BE_{2}(s-1)$ actually contains terms in the time interval
$(2,2\cdot R_{1}]$, then $BE_{2}(s)$ is strictly smaller than $G^s-1$.

We are left to prove the hypothesis for $i=s$ given that there is
at least one checkpoint update in the time interval $(R_{0}=1,R')$. Since
there in no checkpoint in the time interval $(R_{0}=1,R')$ in the snapshot
at $T=2$, then $R'\leq R_{0}+2b=1+2b$. Therefore, $R'/R_{0}\leq1+2b<1/\left(1-b\right)^{2}=G^2$
and by Property~\ref{property:supergeometric2} there is exactly one update
in $(R_{0}=1,R')$. Therefore, the update in $(R_{0}=1,R')$ occurred
at time $R_{1}$, and we denote by $\ell$ the label of the actual
checkpoint involved. Furthermore, we have $R'=R_{2}$.\footnote{The only use of truncating the bounding expression at $T=2$ in the
proof is to limit the number of updates in $(R_{0}=1,R')$ to one.}

Denote by $x$ the time (after $R_{1}$) of the next update of $\ell$
($R_{2}<x\leq2$). After time~$x$, all checkpoints were removed from
$(R_{0}=1,R_{2})$, hence $R_{2}\leq1+b\cdot x$. As in the previous
case, we apply the induction hypothesis from $T=R'=R_{2}$ with $i=s-1$
since we are assured that at least $s-1$ checkpoints are removed from
$(0,R_{2})$ before $R_{2}$ is updated (including $\ell$), hence
the checkpoint updated at $R_{2}$ is $\left(s-1\right)$-stable. We get
$b\cdot BE_{2}(s)\leq b\cdot BE_{2}(s-1)\cdot R_{2}+b\left(R_{2}+G-x\right)$.
Note that we add $b\left(G-x\right)$ to the right hand
side (to bound $b\cdot BE_{2}(s)$) since $\ell$ is first updated
at $R_{1}\le G$ after $T=1$, and not at $x$, which
is the time it is first updated after $R_{2}$ (as considered in $BE_{2}(s-1)$).
Once again, we prove a slightly stronger result than required,
as $BE_{2}(s-1)$ calculated from $T=R_{2}$ may contain terms which
are larger than 2.

Recalling that $R_{2}\leq1+bx=G(1-b)+bx=G+b(x-G)$, we obtain
\begin{align*}
b\cdot BE_{2}(s) & \leq b\cdot BE_{2}(s-1)\cdot R_{2}+b\left(R_{2}+G-x\right)\\
 & =\left(b\cdot BE_{2}(s-1)+b\right)R_{2}-b\left(x-G\right)\\
 & \le\left(G^{s-1}-1+b\right)\left(G+b(x-G)\right)-b\left(x-G\right)\\
 & =(G^{s}-1)(1+b(x/G-1)) -bG(x/G-1)\\
 & =G^{s}-1+b(x/G-1)(G^s-1-G),
\end{align*}
so to show that $b\cdot BE_{2}(s)\le G^s-1$,
it is sufficient to show that $b(x/G-1)(G^s-1-G)\le0$.

Obviously $b>0$; there are 3 checkpoint updates in the time interval $[1,R_{2}]$,
so by Property~\ref{property:supergeometric}, $x>R_{2}>G$, and thus $x/G-1>0$;
lastly, $s\le k/2$ so $G^s\le G^{k/2}\le 2$ by the lemma's assumption, which gives $G^s-1-G\le1-G<0$, as $G>1$.
This completes the induction and the proof of the lemma.
\end{proof}

\begin{corollary}
\label{cor:LB-optimal}For all even $k\ge2$ we have $(1-q_{k}/(k+1))^{-k/2}\ge2$.
In particular, $q_{k}>\ln4$.
\end{corollary}
\begin{proof}
Write $b=q_{k}/(k+1)$ and assume
for the sake of contradiction that $G^{k/2}<2$. By Lemma~\ref{lem:BE-upper-bound}
and Proposition~\ref{prop:BE-at-least-1} we have $1\le b\cdot BE_{2}(k/2)\le G^{k/2}-1$
for a $k/2$-normalized $q_{k}$-efficient $k$-checkpoint algorithm, so $G^{k/2}\ge2$,
contradicting our assumption. Now
\[
\frac{q_{k}}{k+1} = b\ge 1-2^{-2/k}=1-e^{-(\ln4)/k}\ge\frac{\ln4}{k}-\frac{1}{2}\left(\frac{\ln4}{k}\right)^{2} = \left(1-\frac{\ln2}{k}\right)\frac{\ln4}k,
\]
hence $q_{k}\ge\left(1+(1-\ln2)/k-(\ln2)/k^2\right)\ln4 > \ln4$. The last inequality is true when $k>\ln2/(1-\ln2)\approx 2.26$, but we already know that $q_2=1.5>\ln4$.
\end{proof}


\section{\label{sec:applications}Additional Applications of Checkpointing Algorithms}
Most of the applications of online checkpointing algorithms described so far in the literature are related to fault tolerance:
If we discover an error in a lengthy computation, we may want to correct it without restarting the computation from the beginning. In this section
we briefly describe two novel applications for checkpointing algorithms which are motivated by problems in cryptography and cyber security.

Let $f(x)$ be a cryptographic hash function which maps $n$-bit inputs to random-looking $n$-bit outputs. A classical cryptanalytic problem is
to find a collision in a given hash function, i.e., two different inputs $x \neq y$ which are mapped by $f$ to the same output $f(x)=f(y)$.
By the birthday paradox, we expect to find such a collision if we evaluate and compare the values of $f$ for $O(2^{n/2})$ random inputs $x$.
The naive method is to store all these values in an appropriate data structure, but since memory is much more expensive than time, we would like to
find such a repetition using only a small number of $k$ memory cells. For $k=2$ the best known solution is to use Floyd's two finger
algorithm~\cite{Floyd} which iterates the application of $f$ starting from some random initial point $x$. Since the space of values is finite,
the evolving chain must eventually repeat itself, and since $f$ is deterministic the chain will fold into a cycle, and repeat itself forever.
Floyd's algorithm maintains two pointers along the generated chain of values by moving the endpoint pointer at speed~$2$ and
the midpoint pointer at speed~$1$. It stops when the two pointed values are the same. However, this algorithm
is non-optimal for two reasons: It finds a collision only after wasting on average $25 \%$ additional evaluation steps
without noticing that its endpoint is already repeating itself, and it performs on average~$1.5$ evaluations of~$f$ to extend
the evolving chain by one step. To reduce the number of wasted steps, we can maintain a larger number~$k$ of pointers along the evolving
chain, and thus catch the repetition of values at an earlier stage. To make each step more efficient, we can evaluate only the pointer at the
end of the chain, and use our optimal pebbling strategies to leapfrog the memorized pointers to new locations along the chain at zero evaluation cost.
We experimentally tested this strategy with $k=8$ pointers, and observed about $20 \%$ reduction in the worst case waste of our collision finding
algorithm compared to the standard interval doubling algorithm described by Brent~\cite{Brent80}.

A second application is related to backup strategies against sophisticated cyber attacks.
Such attacks try to enhance their destructiveness by stealthily corrupting all the available
backups before launching the actual attack. To model such attacks, we assume that the defender's backup strategy
consists of deciding when to refresh the data in each one of his~$k$ backup devices. When a
backup device is connected to the main computer, one of two things can happen: if the computer is still clean,
the device will instantaneously update all the files with their current contents; if the computer
is already infected, all the data on the device will be lost. The problem is that the defender does not know
whether his computer had already been compromised, and has to prepare for the worst possible choice of infection time.
Note that the standard update strategy of keeping
$k$ external disks stored unpowered in a safe and connecting one of them at the end of each day in a round
robin way will lead to the loss of all the backups $k$ days after the initial infection. It can be shown
that the best backup strategy in this model is to keep all the backups as evenly spread out as possible along the
timeline, so that some of the backup disks will not be connected for a long time, while others will have relatively
fresh versions of the file system. By using our proposed pebbling strategies, the defender can make his worst case loss as small as possible.


\section{\label{sec:conclusion}Concluding Remarks and Open Problems}

In this paper we solved the main open problem in online checkpointing algorithms, 
which is to find tight asymptotic upper and lower bounds on their achievable efficiency.
In addition, we developed efficient techniques for determining tight upper and lower bounds on 
$q_k$ for small values of $k$, which enabled us to develop provably optimal concrete algorithms for 
all $k \leq 10$. However, determining the values of $q_k$ for larger values of $k$ remains a
computationally challenging problem, and finding more efficient ways to compute these values remains 
an interesting open problem.

\appendix

\section{Tables}

Table~\ref{tbl:more-comput-UB} shows the best algorithms our LP approach found for $k=11,12,\ldots,20$. These are (perhaps non-tight) upper bounds on $q_{11}, \ldots, q_{20}$.
Observe how some of the patterns are reminiscent of the pattern used in the algorithm \textsc{Recursive} of Section~\ref{sec:recursive-construction}.

\begin{table}[t]
\centering
\begin{tabular}{|c|c|c|l|c|}
\hline
$k$  & $\lambda$ & $\gamma$  & \multicolumn{1}{c|}{$P$ }  & $n$ \tabularnewline
\hline 
\hline 
11 & 1.4650841 & 8.190656 & (1,3,5,6,1,6,2,10,6,3,6,1,6,2,6,3,9,6) & 18\tabularnewline\hline 
12 & 1.4668421 & 8.862576 & (1,2,3,5,6,7,1,2,6,3,6,7,1,2,6,3,6,9,7) & 19\tabularnewline\hline 
13 & 1.4592320 & 2.94 & (1,3,6,7,4,7,1,7,8,3) & 10\tabularnewline\hline 
14 & 1.4570046 & 58.6 & %
\noindent\begin{minipage}[t]{0.6\columnwidth}%
(1,4,2,6,7,4,7,8,1,8,2,3,7,12,4,7,8,1,4,7,2,7,8,4,13,8,\\
\hphantom{\hspace{0.7ex}}1,8,4,2,7,4,7,8,1,8,4,2,7,12,4,7,13,8)%
\end{minipage} & 44\tabularnewline\hline
15 & 1.4487459 & 2.104027 & (1,2,7,8,4,8,9,5) & 8\tabularnewline\hline 
16 & 1.4487597 & 8.46 & (1,2,4,7,8,9,5,9,1,2,8,4,8,9,5,9,1,2,8,4,8,13,9,5,9) & 25\tabularnewline\hline 
17 & 1.4593611 & 1.694884 & (1,9,5,3,14,8,9) & 7\tabularnewline\hline 
18 & 1.4575670 & 2.57 & (1,8,9,5,9,10,2,5,9,10,3,5) & 12\tabularnewline\hline 
19 & 1.4592194 & 2.45 & (1,9,5,9,10,2,5,9,10,11,3,5) & 12\tabularnewline\hline 
20 & 1.4696048 & 13.3 & %
\noindent\begin{minipage}[t]{0.6\columnwidth}%
(1,5,9,10,2,5,9,10,11,3,5,10,1,5,9,10,11,2,5,10,11,3,5,10,\\
\hphantom{\hspace{0.7ex}}1,5,9,10,6,2,9,10,11,3,5,10)%
\end{minipage} & 36\tabularnewline\hline 
\end{tabular}

\caption{\label{tbl:more-comput-UB}Computationally-verified upper bounds on $q_k$ for $11\le k\le20$.}
\end{table}

Tables~\ref{tbl:recursive-best} and~\ref{tbl:recursive-worst} describe the efficiency of \textsc{Recursive$\left(G^{*},K^{*}\right)$} in two extreme cases:
the ``best'' case $k=2^{m+1}$, which is the special case handled by \textsc{Binary} of~\cite{BDNS13}, and the ``worst'' case $k=2^{m+2}-1$, which shows that the upper bound we proved on the efficiency of \textsc{Recursive} is essentially tight.
\begin{itemize}
\item In the first case, the optimal value $G^*$ for $m\ge7$ is the smallest real root $\neq1$ of 
\[
x^{km/2+k/4}-x^{km/2+k/4-1}-x^{k/2}+1,
\]
i.e.,~\eqref{eq:cond1} and~\eqref{eq:cond2} are tight;
\item In the second case, the optimal value $G^*$ for all $m\ge 0$ is the smallest real root of 
\[
x^{\left(m+1\right)\left(k+1\right)/2}-x^{\left(m+1\right)\left(k+1\right)/2-1}-1,
\]
i.e.,~\eqref{eq:cond1} and~\eqref{eq:cond3} are tight.
\end{itemize}

The second from the right column shows how close $(G^*)^{k/2}$ is to its lower bound $2$, demonstrating the sharpness of Corollary~\ref{cor:LB-optimal}. The rightmost column in each table shows the ``effective constant'' --- defined to be $(qk/((k+1)\ln4)-1)\log_{2}k$ for efficiency $q$ --- as a percentage of the constant $\tau = -\log_2{\ln2}$. The fact that it asymptotically approaches 100\% in Table~\ref{tbl:recursive-worst} shows that indeed $\tau$ is optimal and the analysis is tight.

\begin{table}[t]
\centering
\begin{tabular}{|c|c|c|c|c|c|}
\hline
$k$  & $m$  & Efficiency  & $G^{*}$ is smallest root of  & $\left(G^{*}\right)^{k/2}$  & \%\tabularnewline
\hline
\hline
2  & 0  & 1  & $x-2$  & 2  & $<0$\tabularnewline
\hline
4  & 1  & 1.527864045  & $x^{2}-x-1$  & 2.61803399  & 38.63\% \tabularnewline
\hline
8  & 2  & 1.446619893  & $x^{4}-x-1$  & 2.220744085  & 24.69\% \tabularnewline
\hline
16  & 3  & 1.414522345  & $x^{8}-x-1$  & 2.096981559  & 15.40\% \tabularnewline
\hline
32  & 4  & 1.399982156  & $x^{16}-x-1$  & 2.045751025  & 9.337\% \tabularnewline
\hline
64  & 5  & 1.393037798  & $x^{32}-x-1$  & 2.022250526  & 5.520\% \tabularnewline
\hline
128  & 6  & 1.389641669  & $x^{64}-x-1$  & 2.010975735  & 3.196\% \tabularnewline
\hline
256  & 7  & 1.389039657  & $x^{959}-\left(x^{128}-1\right)/\left(x-1\right)$  & 2.006538067  & 2.996\% \tabularnewline
\hline
512  & 8  & 1.38976776  & $x^{2175}-\left(x^{256}-1\right)/\left(x-1\right)$  & 2.005370202  & 4.265\% \tabularnewline
\hline
1024  & 9  & 1.389961428  & $x^{4863}-\left(x^{512}-1\right)/\left(x-1\right)$  & 2.004616597  & 5.003\% \tabularnewline
\hline
2048  & 10  & 1.389901672  & $x^{10751}-\left(x^{1024}-1\right)/\left(x-1\right)$  & 2.004083324  & 5.413\% \tabularnewline
\hline
4096  & 11  & 1.3897339  & $x^{23551}-\left(x^{2048}-1\right)/\left(x-1\right)$  & 2.003678733  & 5.631\% \tabularnewline
\hline
8192  & 12  & 1.389529892  & $x^{51199}-\left(x^{4096}-1\right)/\left(x-1\right)$  & 2.003356204  & 5.738\% \tabularnewline
\hline
16384  & 13  & 1.389323191  & $x^{110591}-\left(x^{8192}-1\right)/\left(x-1\right)$  & 2.003090123  & 5.785\% \tabularnewline
\hline
32768  & 14  & 1.389128152  & $x^{237567}-\left(x^{k/2}-1\right)/\left(x-1\right)$  & 2.002865287  & 5.799\% \tabularnewline
\hline
65536  & 15  & 1.388949844  & $x^{507903}-\left(x^{k/2}-1\right)/\left(x-1\right)$  & 2.002671985  & 5.796\% \tabularnewline
\hline
131072  & 16  & 1.388789052  & $x^{1081343}-\left(x^{k/2}-1\right)/\left(x-1\right)$  & 2.002503614  & 5.786\% \tabularnewline
\hline
262144  & 17  & 1.388644741  & $x^{2293759}-\left(x^{k/2}-1\right)/\left(x-1\right)$  & 2.002355444  & 5.772\% \tabularnewline
\hline
\end{tabular}

\caption{\label{tbl:recursive-best}Efficiency of \textsc{Recursive$\left(G^{*},K^{*}\right)$}
for $k=2^{m+1}$.}
\end{table}
\begin{table}[h]
\centering
\begin{tabular}{|c|c|c|c|c|c|}
\hline
$k$  & $m$  &  efficiency  & $G^{*}$ is smallest root of  & $\left(G^{*}\right)^{k/2}$  & \% \tabularnewline
\hline
\hline
3  & 0  & 1.145898034  & $x^{2}-x-1$  & 2.058171027  & $<0$\tabularnewline
\hline
7  & 1  & 1.318433761  & $x^{8}-x^{7}-1$  & 2.075892596  & $<0$\tabularnewline
\hline
15  & 2  & 1.408092224  & $x^{24}-x^{23}-1$  & 2.09450451  & 11.62\%\tabularnewline
\hline
31  & 3  & 1.448810165  & $x^{64}-x^{63}-1$  & 2.099878619  & 42.25\% \tabularnewline
\hline
63  & 4  & 1.46399549  & $x^{160}-x^{159}-1$  & 2.097270918  & 63.36\% \tabularnewline
\hline
127  & 5  & 1.466865403  & $x^{384}-x^{383}-1$  & 2.091122952  & 76.82\% \tabularnewline
\hline
255  & 6  & 1.464278319  & $x^{896}-x^{895}-1$  & 2.083917032  & 85.05\% \tabularnewline
\hline
511  & 7  & 1.459602376  & $x^{2048}-x^{2047}-1$  & 2.076835761  & 89.98\% \tabularnewline
\hline
1023  & 8  & 1.454408143  & $x^{4608}-x^{4607}-1$  & 2.070357915  & 92.91\% \tabularnewline
\hline
2047  & 9  & 1.449377613  & $x^{10240}-x^{10239}-1$  & 2.064618545  & 94.66\% \tabularnewline
\hline
4095  & 10  & 1.444769023  & $x^{22528}-x^{22527}-1$  & 2.059600382  & 95.72\% \tabularnewline
\hline
8191  & 11  & 1.440647199  & $x^{49152}-x^{49151}-1$  & 2.055228333  & 96.39\% \tabularnewline
\hline
16383  & 12  & 1.436994729  & $x^{106496}-x^{106495}-1$  & 2.051413108  & 96.83\% \tabularnewline
\hline
32767  & 13  & 1.43376402  & $x^{229376}-x^{229375}-1$  & 2.048069607  & 97.14\% \tabularnewline
\hline
65535  & 14  & 1.430900622  & $x^{491520}-x^{491519}-1$  & 2.045123383  & 97.36\% \tabularnewline
\hline
131071  & 15  & 1.428352881  & $x^{1048576}-x^{1048575}-1$  & 2.042511814  & 97.54\% \tabularnewline
\hline
262143  & 16  & 1.426075306  & $x^{2228224}-x^{2228223}-1$  & 2.040183169  & 97.69\% \tabularnewline
\hline
\end{tabular}

\caption{\label{tbl:recursive-worst}Efficiency of \textsc{Recursive$\left(G^{*},K^{*}\right)$}
for $k=2^{m+2}-1$.}
\end{table}

\begin{acks}

The work was partially supported by the 
the European Research Council under the ERC starting grant agreement n.~757731 (LightCrypt), the BIU Center for Research in Applied Cryptography and Cyber Security in conjunction with the Israel National Cyber Bureau in the Prime Minister's Office, and
by the Israeli Science Foundation through grant No.~573/16.
\end{acks}

\bibliographystyle{ACM-Reference-Format}
\bibliography{Checkpoints}

\end{document}